\documentclass[pra,10pt,twocolumn,superscriptaddress,floatfix,nofootinbib,showpacs]{revtex4-2}
\usepackage{amsmath,amssymb,amsthm}
\usepackage{array}
\usepackage{multirow}
\usepackage{easybmat}
\usepackage[colorlinks=true,citecolor=blue,urlcolor=blue]{hyperref}
\usepackage[pdftex]{graphicx}
\usepackage{float}
\usepackage{times,txfonts}
\usepackage{braket}
\usepackage{color}
\usepackage{natbib}
\usepackage[english]{babel}
\usepackage{hyperref}
\usepackage{mathrsfs}
\usepackage{subfig}
\usepackage{caption}

\newtheorem{theorem}{Theorem}
\newtheorem{lemma}{Lemma}

\begin{document}
\title{Efficient utilization of imaginarity in quantum steering}  
\author{Shounak Datta}
\email{shounaknew@gmail.com}
\affiliation{S. N. Bose National Centre for Basic Sciences, Block JD, Sector III, Salt Lake, Kolkata 700 106, India}
\author{A. S. Majumdar}
\email{archan@bose.res.in}
\affiliation{S. N. Bose National Centre for Basic Sciences, Block JD, Sector III, Salt Lake, Kolkata 700 106, India}
\begin{abstract}
We illustrate the role of complex numbers in quantum information processing through the phenomenon of quantum steering. Exploiting partial knowledge of a qubit in terms of imaginarity, we formulate a steering criterion for bipartite qubit systems, which requires two dichotomic measurements at the untrusted side and two mutually unbiased bases at the trusted side. We show that quantum correlations embodied through our proposed imaginarity steering inequality can be witnessed by suitably constructed Hermitian operators that depend on fewer state parameters and require measurement of a lesser number of observables than for witnessing other forms of quantum nonlocal correlations. The monogamy of such correlations is also demonstrated.  We further underscore the steering of imaginarity as an efficient nonlocal resource in the presence of white noise and under unhsarp measurements, showing the  robustness of imaginarity steering compared to certain other steering criteria.
\end{abstract}

\maketitle 

\section{Introduction}\label{C1}

Real numbers are crucial for explaining most of the features in our classical world. Further, real vector space quantum theory has been successfully applied to investigate entanglement between real qubits(rebits)~\citep{Batle,Wootters1,Wootters2,Prasannan} and applied to some information processing tasks~\citep{McKague,Koh}. However, quantum mechanics describes the state of a quantum system by wave functions having the form $\Psi = a e^{i \Phi}$ where complex numbers play a key role in many aspects. If a physical theory is constructed with the rules of quantum mechanics although restricted to real parameters in terms of states and operations~\citep{Araki,Hardy,Baez,Aleksa,Wootters},  it fails to explain certain feasible properties~\citep{Liu,Zhu,Zhang}.  

The importance of imaginary numbers in quantum theory was first proposed by Hickey et al~\citep{Hickey} from the perspective of general quantum resource theory. The operational resource theory of imaginarity has been put forward by Wu et al~\citep{Wu}, where the importance of imaginarity in experimental scenario and in solving the local state discrimination problem has been discussed. Imaginarity can be quantified and applied to state conversion problems~\citep{Wu1}, and the measures, such as, geometric imaginarity and robustness of imaginarity turn out to be useful in linear optical set-ups. Resource theory of imaginarity has been extended in the distributed scenario~\citep{Wu2}. Based on the resource theory, it has been shown that real operations can broadcast imaginarity  like other quantum mechanical resources~\citep{Zhang2024}. Imaginarity is also useful for distinguishing real quantum states which can be used to hide or mask quantum information~\citep{Zhu2021}. Further, real and complex quantum theories can be discriminated~\citep{Batle2025} by employing imaginarity in nonlocality experiments to show the supremacy of complex theories over the real world-view.

Imaginarity is manifested through density matrix representation of a quantum system in  particular basis. Since mutually orthogonal bases are critical in the development of quantum nonlocal theory,  it is interesting to investigate the role of imaginarity in quantum nonlocality. Historically, the Einstein-Podolsky-Rosen(EPR) paradox~\citep{EPR} and Bell's inequality~\citep{Bell,CHSH}, first revealed the nonlocal nature of correlations in quantum mechanics for entangled pairs. Schr\"{o}dinger's interpretation~\citep{Sch} introduced a new kind of nonlocal phenomenon, called quantum steering. There exists a necessary and sufficient condition for steering in 2-2-2 scenario~\citep{CFFW}. Steerability is an intermediate form of quantum correlation lying between Bell nonlocality and entanglement~\citep{Wiseman,Jones}. Quantum steering has emerged as a powerful resource in many different information processing tasks such as 
semi-device-independent self-testing of quantum states and devices~\citep{Supic2020,Bian2020}.

The formulation of steering through inequalities is based on the uncertainty between mutually orthogonal basis in different forms~\citep{Uola}. The expression of  steering inequalities~\citep{Reid,Schneeloch,Pramanik,Maity} through various uncertainty relations, such as Heisenberg uncertainty relation~\citep{Heisenberg}, entropic uncertainty relation~\citep{Massen}, fine-grained uncertainty relation~\citep{Oppenheim}, and sum uncertainty relation~\citep{Pati}, requires complete description of the local density matrix. On the other hand, there exist  certain quantum mechanical correlations  based on complementarity in mutually orthogonal basis, which can be revealed using partial knowledge of the density matrix~\citep{Mondal,Wei2024}. The motivation for the present work is to extend such an approach by utilizing the resource of imaginarity for formulating an efficient criterion for quantum steering, that does not require full density matrix information, and is hence, easier to implement in practice.  

To this end, here we derive an imaginarity steering inequality based on partial information of the density matrix, and propose the following protocol for quantum steering, which can be easily realized with minimal experimental resources.  Consider that, Alice prepares a bipartite two-qubit state and shares the second particle with Bob. Alice performs a measurement on her particle in order to convince Bob that the joint state is entangled. Next, Bob determines the imaginarity of  his local qubit and checks our steering criterion of imaginarity, through the measurement statistics of two dichotomic measurements at Alice's end and two mutually orthogonal basis at Bob's end. If our imaginarity criterion is violated, Bob becomes convinced that Alice has not cheated him and he does not possess a local hidden state (LHS) of a qubit. In essence, the control of Alice over Bob's local imaginarity can be considered as a sufficient criterion to detect nonlocality.

In this work, we apply the robustness of imaginarity as a quantifier of imaginarity~\citep{Hickey,Wu1}. We construct a complementarity relation of imaginarity involving two mutually unbiased bases in $\mathbb{C}^2$~\citep{MUB}. Subsequently, we construct our 2-measurement imaginarity steering inequality (ISI) for bipartite qubit states. We show that the untrusted side of an entangled pair of qubits can steer the quantum imaginarity of the trusted side when ISI is violated. We next show that the set of states satisfying ISI is convex and compact, enabling construction of a witness operator to detect the steering of imaginarity.  The steerability of imaginarity is exemplified through  well-known families of bipartite qubit states~\citep{Werner,Rau2009}. Moreover, by introducing a general tripartite pure state~\citep{Acin2000}, we show that a strong monogamy trade-off~\citep{Dhar2017} exists between its bipartite cuts in terms of ISI. This result implies that two parties can not simultaneously steer the imaginarity of a third party. 

Further, we demonstrate the efficacy of our imaginarity steering relation in comparison with similar other steering criteria which are based on partial knowledge of the state. Specifically, we show that the 2-measurement set-up imaginarity steering inequality proposed in our work is more robust in revealing quantum steering compared to the formalisms based on 3-measurement set-ups~\citep{Mondal, Wei2024}. The imaginarity steering criterion is able to tolerate more noise, and is also more resilient if unsharpness of measurements is introduced. Our analysis shows that the manifestation of ISI not only requires lesser experimental resources in order to determine a fewer number of state parameters, but is also more robust against various noise effects. 

We organise this article as follows. In Section \ref{C2} we present a  brief overview on certain preliminaries of quantum steering. In Section \ref{C3}, we formulate  our imaginarity steering inequality with the help of complementarity relations of imaginarity among two mutually unbiased bases. In Section \ref{C4}, we show how imaginarity revealed through ISI can be detected through the witness operator framework, and also provide examples. In  Section \ref{C5}, we demonstrate the monogamy of ISI in the tripartite scenario. The effectiveness of ISI  under two distinct scenarios of noise and unsharp measurements is revealed  in Section \ref{C6}. Finally, we conclude with a summary in Section \ref{C7}. 

\section{Preliminaries of quantum steering}\label{C2}

Suppose that Alice prepares a joint state and keeps one particle with her and sends the other to Bob. Now, the task of steering Bobs's system by Alice is to convince Bob that the shared state is entangled. Depending upon Alice's measurement on her particle, if the measurement statistics at Bob's side can not be determined from a pre-fixed cheating strategy considered by Alice, then Bob is convinced that the shared state is entangled and his state can be steered or controlled by Alice. The cheating strategy may come from a Local Hidden State(LHS) model~\citep{Wiseman,Jones} given by
\begin{equation}
\sigma_{a|A} = \sum_{\lambda} p(\lambda) ~p(a|A,\lambda) ~\rho_{\lambda}
\label{CS}
\end{equation}
where the distribution of local hidden variable $\lambda$ satisfies $\sum_{\lambda} p(\lambda)=1$ and the probability of getting outcome $a$ from measurement $A$ for a given $\lambda$ at Alice's side is $p(a|A,\lambda)$. The unnormalised conditional state at Bob's side  can be written as, $\sigma_{a|A} = p(a|A) ~\rho_{a|A}$, where $p(a|A)$ is the probability of getting outcome $a$ corresponding to Alice's measurement $A$ and $\rho_{a|A}$ is the normalised conditional state at Bob's side. The LHS corresponding to Bob is given by $\rho_{\lambda}$. As the system of Bob is assumed to be quantum,  the joint probability distribution of getting outcomes $\lbrace a, b \rbrace$ from measurements $\lbrace A, B \rbrace$ corresponding to Alice and Bob respectively, is governed by the Born rule, as follows.
\begin{equation}
p(a,b|A,B) = \sum_{\lambda} p(\lambda) ~p(a|A,\lambda) ~p(b|B,\rho_{\lambda})
\label{LHS}
\end{equation}
where $p(b|B,\rho_{\lambda})=\operatorname{Tr}[\Pi_{b|B} ~\rho_{\lambda}]$ and $p(a,b|A,B)=\operatorname{Tr}[\Pi_{b|B} ~\sigma_{a|A}]$ and $\Pi_{b|B}$ is the projective measurement operator which yields outcome $b$ from measurement $B$ at Bob's side.

The non-existence of the LHS model through the violation of Eq.(\ref{LHS}) is demonstrated through the violation of various steering inequalities. There exists a necessary and sufficient condition for steering of quantum states in the 2-2-2 scenario~\citep{CFFW}, i.e.
\begin{align}
&\sqrt{\langle (A_1 + A_2) B_1 \rangle^2 + \langle (A_1 + A_2) B_2 \rangle^2} ~+ \nonumber\\
&\sqrt{\langle (A_1 - A_2) B_1 \rangle^2 + \langle (A_1 - A_2) B_2 \rangle^2} \leq 2
\label{CFFW}\end{align}
where the dichotomic observables $\lbrace A_1,A_2 \rbrace$ correspond to Alice and the mutually unbiased observables $\lbrace B_1,B_2 \rbrace$ correspond to Bob. The nonlinear inequality(\ref{CFFW}) is known as analogous Clauser-Horne-Shimony-Holt(CHSH) inequality for steering or CFFW inequality.

\section{Steering of quantum imaginarity}\label{C3}

The mathematical structure of quantum mechanics mainly depends on complex numbers, but the verifiable quantities through physical operations are predominantly real. The role of the imaginary part of the elements of a density matrix lead to interesting consequences in quantum information theory. The resource theoretic framework of imaginarity has been developed and imaginarity  quantified as a convex monotone~\citep{Wu,Wu1}. Imaginarity is particularly significant in local state discrimination, since bipartite real mixed states which are perfectly distinguishable via local operations and classical communication (LOCC), can not be distinguished via local real operations and classical communication (LRCC). Imaginarity, alike quantum coherence~\citep{Baumgratz,Girolami,Winter}, is a basis dependent property of a quantum state.

A quantifier of imaginarity, called the robustness of imaginarity, similar to distance-based measures~\citep{Brandao}, minimizes the distance between a given state (say $\rho$) and a state from the set of real states, $\mathcal{R}$. It is defined as~\citep{Hickey,Wu1},
\begin{equation}
\mathscr{I}_R (\rho) = \min_{\tau} \lbrace s\geq 0 : \frac{\rho + s \tau}{1+s} \in \mathcal{R} \rbrace
\end{equation}
by considering minimization over all quantum states $\tau$. Interestingly, $\mathscr{I}_R (\rho)$ has a closed and simplified form for all pure and mixed states. It is expressed as,
\begin{equation}
\mathscr{I}_R (\rho) = \frac{1}{2} \Vert \rho - \rho^T \Vert_1
\label{ROI}
\end{equation}
where, $\rho^T$ is the transposition of $\rho$ and the trace norm $\Vert V \Vert_1 = \operatorname{Tr} [\sqrt{V^{\dagger} V}]$. If $\rho$ is pure, i.e. $\rho=|\psi\rangle\langle\psi|$, then $\mathscr{I}_R (\rho)=\sqrt{1-|\langle \psi^*|\psi \rangle|^2}$. As a bonafide measure, it is a convex monotone under local real operations, and vanishes identically for all real states.

\subsection{Imaginarity complementarity relations}

Consider a qubit state, $\rho= \frac{\openone_2 +\vec{n}.\vec{\sigma}}{2}$, where $\vec{n}=(n_x,n_y,n_z)$ with $0\leq |\vec{n}|\leq 1$, and $\vec{\sigma}$ the Pauli vector. The three mutually unbiased bases in $\mathbb{C}^2$~\citep{MUB} can be generally represented as $B_{1} \equiv \lbrace |u_{\theta,\phi}\rangle, |d_{\theta,\phi}\rangle \rbrace$, $B_{2} \equiv\lbrace \frac{|u_{\theta,\phi}\rangle +|d_{\theta,\phi}\rangle}{\sqrt{2}}, \frac{|u_{\theta,\phi}\rangle -|d_{\theta,\phi}\rangle}{\sqrt{2}} \rbrace$ and $B_{3} \equiv\lbrace \frac{|u_{\theta,\phi}\rangle +i|d_{\theta,\phi}\rangle}{\sqrt{2}}, \frac{|u_{\theta,\phi}\rangle -i|d_{\theta,\phi}\rangle}{\sqrt{2}} \rbrace$,  where  $|u_{\theta,\phi}\rangle = \cos\frac{\theta}{2}|0\rangle + e^{i\phi} \sin\frac{\theta}{2} |1\rangle$, $|d_{\theta,\phi}\rangle = \sin\frac{\theta}{2}|0\rangle - e^{i\phi} \cos\frac{\theta}{2} |1\rangle$, $0\leq \theta \leq \pi$ and $0\leq \phi \leq 2\pi$. The modulus of inner product of a pair of bases chosen from two different MUBs always yields the value of $\frac{1}{\sqrt{2}}$ and each basis from $\lbrace B_1,B_2,B_3 \rbrace$ constitutes set of orthonormal vectors. By representing $\rho$ in the above mentioned MUBs (e.g. $\rho_{B_{1}} = \begin{pmatrix}
\langle u_{\theta,\phi}|\rho|u_{\theta,\phi} \rangle & \langle u_{\theta,\phi}|\rho|d_{\theta,\phi} \rangle \\
\langle d_{\theta,\phi}|\rho|u_{\theta,\phi} \rangle & \langle d_{\theta,\phi}|\rho|d_{\theta,\phi} \rangle
\end{pmatrix}$ in $B_{1}$ basis), we can find from Eq.(\ref{ROI}) that,
\begin{eqnarray}
&\mathscr{I}_R^{B_1}(\rho) = \mathscr{I}_R^{B_2}(\rho) = |n_y \cos\phi -n_x \sin\phi| \nonumber\\
&\mathscr{I}_R^{B_3}(\rho) = |n_x \cos\theta \cos\phi + n_y \cos\theta \sin\phi -n_z \sin\theta|.
\end{eqnarray} 
The uncertainty between unequal weights of imaginarities among different MUBs leads to a complementarity relation that can be minimally demonstrated by considering two MUBs, and is given by
\begin{equation}
\mathscr{I}_R^{B_2(B_1)}(\rho) + \mathscr{I}_R^{B_3}(\rho) \leq \sqrt{2}
\end{equation}
where the upper bound corresponds to $\theta=0$ and $\phi=\pi$. The choice of MUBs, i.e., $\lbrace \sigma_x, \sigma_y, \sigma_z \rbrace$-basis is  optimal.

The density operator of $\rho$ can be represented as $\begin{pmatrix}
\frac{1+n_x}{2} & \frac{i n_y + n_z}{2} \\
\frac{-i n_y + n_z}{2} & \frac{1-n_x}{2}
\end{pmatrix}$, $\begin{pmatrix}
\frac{1+n_y}{2} & \frac{-i n_x + n_z}{2} \\
\frac{i n_x + n_z}{2} & \frac{1-n_y}{2}
\end{pmatrix}$ and $\begin{pmatrix}
\frac{1+n_z}{2} & \frac{n_x - i n_y}{2} \\
\frac{n_x + i n_y}{2} & \frac{1-n_z}{2}
\end{pmatrix}$ corresponding to x-, y- and z-basis respectively. In terms of the above representations, we observe from Eq.(\ref{ROI}) that,
\begin{eqnarray}
&\mathscr{I}_R^x(\rho) = \mathscr{I}_R^z(\rho) = n_y \nonumber\\
&\mathscr{I}_R^y(\rho) = n_x
\end{eqnarray}
This shows the feature of basis dependence in terms of the robustness of imaginarity. A complementarity relation between two mutually orthogonal basis x and y ca be similarly constructed as,
\begin{equation}
\mathscr{I}_R^x(\rho) + \mathscr{I}_R^y(\rho) \leq \sqrt{2}
\label{ICR1}
\end{equation}
The equality of Eq.(\ref{ICR1}) holds for the state $\rho= \frac{\openone_2 + \frac{1}{\sqrt{2}} (\sigma_x + \sigma_y)}{2}$, while  the left hand sides of Eq.(\ref{ICR1}) becomes 1 for the eigenstates of $\sigma_y$.

\subsection{ Imaginarity steering inequality}

Like coherence, imaginarity also explores a portion of knowledge hidden in a quantum state. As the properties are  different from each other, hence a unique steering criterion can be formulated using imaginarity. To formulate a steering criterion, a LHS model has to be set up based on imaginarity complementarity relations.

The cheating strategy given by Eq.(\ref{CS}) can be recast for measurement along $A$-direction and outcome $a$ at Alice's side in terms of imaginarity in any basis at Bob's side as,
\begin{equation}
\mathscr{I}_R(\sigma_{a|A}) = \sum_{\lambda} p(\lambda) ~p(a|A,\lambda) ~\mathscr{I}_R(\rho_{\lambda})
\label{CS1}
\end{equation}
by using the facts that, $\lbrace p(\lambda), p(a|A,\lambda) \rbrace \in \mathbb{R} ~\forall \lambda$ and matrix norms are absolutely homogeneous of degree 1. As $0\leq \lbrace \mathscr{I}_R(\rho_{\lambda}), p(\lambda), p(a|A,\lambda) \rbrace \leq 1$ for all local hidden variable $\lambda$ and $p(a|A) \in \mathbb{R}$, it follows from from Eq.(\ref{CS1}) that,
\begin{align}
\mathscr{I}_R(\sigma_{a|A}) &= p(a|A) ~\mathscr{I}_R(\rho_{a|A}) \nonumber\\
&\leq \max_{\lambda} \mathscr{I}_R(\rho_{\lambda}) ~[\sum_{\lambda} p(\lambda) ~p(a|A,\lambda)] \nonumber\\
&= \mathscr{I}_R(\rho_{\lambda_{\max}}) ~p(a|A)
\label{I1}
\end{align}
Using $\sum_{a=+,-} p(a|A) =1$, we can write from Eq.(\ref{I1}) that,
\begin{align}
\sum_{a=+,-} p(a|A) ~\mathscr{I}_R(\rho_{a|A}) &\leq \mathscr{I}_R(\rho_{\lambda_{\max}}) ~\sum_{a=+,-} p(a|A) \nonumber\\
&= \mathscr{I}_R(\rho_{\lambda_{\max}})
\end{align}

If Bob computes imaginarity in  x-basis after Alice's measurement in y-basis with outcome $a_1$, then we have
\begin{equation}
\sum_{a_1=+,-} p(a_1|y) ~\mathscr{I}_R^x(\rho_{a_1|y}) \leq \mathscr{I}_R^x(\rho_{\lambda_{\max}})
\label{x} 
\end{equation}

Similarly, if Bob calculates imaginarity in y-basis after Alice's measurement in x-basis with outcome $a_2$, then 
\begin{equation}
\sum_{a_2=+,-} p(a_2|x) ~\mathscr{I}_R^y(\rho_{a_2|x}) \leq \mathscr{I}_R^y(\rho_{\lambda_{\max}}) 
\label{y}
\end{equation}

Now by summing inequalities (\ref{x}) and (\ref{y}) and using Eq.(\ref{ICR1}) subsequently, we obtain
\begin{align}
I_2= \sum_{a_1=+,-} p(a_1|y) ~\mathscr{I}_R^x(\rho_{a_1|y}) + \sum_{a_2=+,-} p(a_2|x) ~\mathscr{I}_R^y(\rho_{a_2|x}) \leq \sqrt{2}
\label{ISI1}
\end{align}

We have constructed the inequality (\ref{ISI1}) by modelling LHS for Bob through the imaginarity complementarity relation. The violation of it thus implies steering of Bob's local imaginarity by Alice. Hence we call the sufficient criteria given by inequalities (\ref{ISI1}) as \textit{imaginarity steering inequalities (ISI)} under the 2-measurement scenario. It certifies that imaginarity can be an aid in quantum nonlocality where full information about LHS is not required.

\begin{theorem}
All separable states satisfy imaginarity steering inequalities.
\end{theorem}

\begin{proof}
A general bipartite separable state has the following form,
\begin{equation}
\rho_{AB}^{\text{sep}} = \sum_{\lambda} p(\lambda) ~\rho_{\lambda}^A \otimes \rho_{\lambda}^B
\end{equation}
Now depending upon Alice's projective measurement $A$ and outcome $a$, Bob's normalized conditional state becomes,
\begin{align}
\rho_{a|A} &= \frac{\sum_{\lambda} p(\lambda) ~\operatorname{Tr}_A [(\Pi_{a|A} \otimes \openone_2)~(\rho_{\lambda}^A \otimes \rho_{\lambda}^B)]}{\sum_{\lambda} p(\lambda) ~\operatorname{Tr} [(\Pi_{a|A} \otimes \openone_2)~(\rho_{\lambda}^A \otimes \rho_{\lambda}^B)]} \nonumber\\
&= \frac{\sum_{\lambda} p(\lambda) ~\operatorname{Tr}_A [\Pi_{a|A} ~\rho_{\lambda}^A] ~\rho_{\lambda}^B}{\sum_{\lambda} p(\lambda) ~\operatorname{Tr}_A [\Pi_{a|A} ~\rho_{\lambda}^A] ~\operatorname{Tr}_B [\rho_{\lambda}^B]} \nonumber\\
&= \frac{\sum_{\lambda} p(\lambda) ~p(a|A,\lambda) ~\rho_{\lambda}^B}{\sum_{\lambda} p(\lambda) ~p(a|A,\lambda)} \nonumber\\
&= \frac{\sum_{\lambda} p(\lambda) ~p(a|A,\lambda) ~\rho_{\lambda}^B}{p(a|A)}
\label{sep}
\end{align} 

Since $\lbrace p(\lambda), p(a|A,\lambda), p(a|A) \rbrace \in \mathbb{R} ~\forall \lambda$, therefore using subadditivity and homogeneity properties of the trace norm of a matrix, i.e., $\mathscr{I}_R(\sum_i \varrho_i) \leq \sum_i \mathscr{I}_R(\varrho_i)$ and $\mathscr{I}_R(c \varrho) = c \mathscr{I}_R(\varrho)$ respectively, where $c \in \mathbb{R}$, and summing over all outcomes at Alice's side, we can write from Eq.(\ref{sep}) that,
\begin{align}
\sum_{a=+,-} p(a|A) ~\mathscr{I}_R(\rho_{a|A}) &\leq \sum_{a,\lambda} p(\lambda) ~p(a|A,\lambda) ~\mathscr{I}_R(\rho_{\lambda}^B) \nonumber\\
&= \sum_{\lambda} p(\lambda) ~\mathscr{I}_R(\rho_{\lambda}^B)
\label{sep1}
\end{align}
where we apply $\sum_a p(a|A,\lambda) =1$.

Now by employing Eq.(\ref{sep1}) in different bases as given in the inequality(\ref{ICR1}), the left hand side of ISI(\ref{ISI1}) becomes,
\begin{align}
I_2(\rho_{AB}^{\text{sep}}) &\leq \sum_{\lambda} p(\lambda) ~(\mathscr{I}_R^x(\rho_{\lambda}^B) + \mathscr{I}_R^y(\rho_{\lambda}^B)) \nonumber\\
&\leq \sqrt{2} ~\sum_{\lambda} p(\lambda) = \sqrt{2}
\label{sep2}
\end{align}
where we use $\sum_{\lambda} p(\lambda)=1$. As a result, we observe that, the violation of ISI requires the composite system to be entangled.
\end{proof}

A general bipartite qubit state is represented as~\citep{Fano},
\begin{equation}
\rho_{AB} = \frac{1}{4} [\openone_2 \otimes \openone_2 + \vec{m}.\vec{\sigma} \otimes \openone_2 + \openone_2 \otimes \vec{n}.\vec{\sigma} + \sum_{i,j=1}^3 t_{ij} \sigma_i \otimes \sigma_j]
\label{gen}
\end{equation}
where $\openone_2$ is $2\times 2$ identity matrix, $\lbrace \sigma_i \rbrace_{i=1}^3$ are Pauli spin matrices, the vectors $\lbrace \vec{m}, \vec{n} \rbrace \in \mathbb{R}^3$ ($|\vec{m}| \leq 1$, $|\vec{n}| \leq 1$) correspond to Alice and Bob respectively and the elements $t_{ij} = \operatorname{Tr}[(\sigma_i \otimes \sigma_j) \rho_{AB}]$ form $3\times 3$ correlation matrix $T=\lbrace t_{ij} \rbrace_{i,j=1}^3$. Eq.(\ref{gen}) can be chosen as the most general test set of density matrices to demonstrate the steering of imaginarity at Bob's side. In general, the detection of quantum nonlocal correlations, {\it viz.}, Bell nonlocality and quantum steering is determined by 9 out of 15 real parameters of the correlation matrix, $T$ as obtained from Eq.(\ref{gen})~\citep{Horodecki1995,Horodecki1996}. However, our 2-measure criterion for steering imaginarity depends upon the limited parametric space of the function $I_2(\rho_{AB})$ as discussed below.
\begin{theorem}
The imaginarity of Bob's local state can be steered by Alice with the use of generalised two-qubit state given by Eq.(\ref{gen}) when
\begin{equation}
I_2(\rho_{AB})=\frac{1}{2}(|n_1-t_{11}|+|n_1+t_{11}|+|n_2-t_{22}|+|n_2+t_{22}|)>\sqrt{2}
\end{equation} where the nonlocality can be demonstrated by using the above 4 real parameters only.
\end{theorem}
\begin{proof}
By using the general parametric representation of a two-qubit state $\rho_{AB}$ given by Eq.(\ref{gen}), we have $p(\pm|x)=\operatorname{Tr}[\big(\frac{\openone_2 \pm \sigma_x}{2}\otimes \openone_2\big) ~\rho_{AB}] = \frac{1 \pm m_1}{2}$ and $p(\pm|y)=\operatorname{Tr}[\big(\frac{\openone_2 \pm \sigma_y}{2}\otimes \openone_2\big) ~\rho_{AB}] = \frac{1 \pm m_2}{2}$ respectively. Bob computes imaginarity of his local state depending upon the measurement done at Alice's side. If Alice measures in x-basis, then imaginarity of Bob's local state in y-basis becomes
\begin{equation}
\mathscr{I}_R^y(\rho_{\pm|x}) = \frac{|n_1 \pm t_{11}|}{1 \pm m_1}
\label{Ix}
\end{equation}
corresponding to outcomes $\pm$ at Alice's side.
On the other hand, if Alice measures in y-basis, the imaginarity of Bob's local state in x-basis becomes
\begin{equation}
\mathscr{I}_R^x(\rho_{\pm|y}) = \frac{|n_2 \pm t_{22}|}{1 \pm m_2}
\label{Iy}
\end{equation}
corresponding to outcomes $\pm$ at Alice's side. By using Eq.(\ref{Ix}) and Eq.(\ref{Iy}) in the inequality(\ref{ISI1}), it follows that
\begin{equation}
I_2(\rho_{AB})=\frac{1}{2}(|n_1-t_{11}|+|n_1+t_{11}|+|n_2-t_{22}|+|n_2+t_{22}|)
\label{IGen}
\end{equation}
which depends on 4 real parameters $n_1$, $n_2$, $t_{11}$ and $t_{22}$ only, where $0 \leq \lbrace |n_1|,|n_2|,|t_{11}|,|t_{22}| \rbrace \leq 1$. Hence, nonlocality can be demonstrated through the violation of inequality(\ref{ISI1}), i.e. $I_2(\rho_{AB}) > \sqrt{2}$ with the use of 4 real parameters of the state given by Eq.(\ref{gen}) which shows that it needs smaller number of state parameters to demonstrate nonlocality compared to that required for other forms of nonlocality. Since our formulation requires lesser tomographic knowledge of the state, it promises experimental advantage for determination of the non-classicality.
\end{proof}

\section{Witnessing Steering of Imaginarity}\label{C4}

In this section, we show that the formalism of witness operators can be applied to identify the quantum states which obey our imaginarity steering relation. To this end, we first prove that the set of states obeying the ISI (\ref{ISI1}) is both convex and compact.

\begin{theorem}
The set of density matrices $\mathcal{S}=\lbrace\rho_{AB}: I_2(\rho_{AB})\leq \sqrt{2}\rbrace$ form a convex and compact set. 
\end{theorem}
\begin{proof}
Let us take any two density matrices $\rho_{AB}^{(1)}$ and $\rho_{AB}^{(2)}$ such that $\rho_{AB}^{(1)},\rho_{AB}^{(2)} \in \mathcal{S}$, {\it i.e.},
\begin{eqnarray}
I_2(\rho_{AB}^{(1)})\leq \sqrt{2}, \\
I_2(\rho_{AB}^{(2)})\leq \sqrt{2}
\label{Irho12}
\end{eqnarray}
Let us consider that a density matrix, $\rho_{AB}^{(c)}$ can be expressed as a convex combination of $\rho_{AB}^{(1)}$ and $\rho_{AB}^{(2)}$, implying
\begin{equation}
\rho_{AB}^{(c)} = \alpha \rho_{AB}^{(1)} + (1-\alpha) \rho_{AB}^{(2)}
\label{rhoc}
\end{equation}
where $0 \leq \alpha \leq 1$. To prove convexity of the set $\mathcal{S}$, we have to show that
\begin{equation}
I_2(\rho_{AB}^{(c)}) \leq \alpha I_2(\rho_{AB}^{(1)}) + (1-\alpha) I_2(\rho_{AB}^{(2)}), ~\forall \rho_{AB}^{(c)}
\label{con}
\end{equation}
In respect of the density matrices $\rho_{AB}^{(k)} ~\forall k\in \lbrace 1,2 \rbrace$, by using Eq.(\ref{IGen}), we can write 
\begin{equation}
I_2(\rho_{AB}^{(k)})=\frac{1}{2}(|n_1^{(k)}-t_{11}^{(k)}|+|n_1^{(k)}+t_{11}^{(k)}|+|n_2^{(k)}-t_{22}^{(k)}|+|n_2^{(k)}+t_{22}^{(k)}|), ~\forall k\in \lbrace 1,2 \rbrace
\end{equation}
For the state $\rho_{AB}^{(c)}$ given by Eq.(\ref{rhoc}),   we have
\begin{align}
I_2(\rho_{AB}^{(c)})= \frac{1}{2} \Big[ &|\alpha ~(n_1^{(1)}-t_{11}^{(1)}) + (1-\alpha) ~(n_1^{(2)}-t_{11}^{(2)})| \nonumber\\
& +|\alpha ~(n_1^{(1)}+t_{11}^{(1)}) + (1-\alpha) ~(n_1^{(2)}+t_{11}^{(2)})| \nonumber\\
& +|\alpha ~(n_2^{(1)}-t_{22}^{(1)}) + (1-\alpha) ~(n_2^{(2)}-t_{22}^{(2)})| \nonumber\\
& +|\alpha ~(n_2^{(1)}+t_{22}^{(1)}) + (1-\alpha) ~(n_2^{(2)}+t_{22}^{(2)})| \Big]
\end{align}
By applying the triangle inequality, $|z_1 + z_2| \leq |z_1| + |z_2|$ and $|w z| = w |z|$ for any $\lbrace z_1,z_2,z,w \rbrace \in \mathbb{R}^3$ we can rewrite it as,
\begin{align}
I_2(\rho_{AB}^{(c)}) \leq &\frac{1}{2} \Big[ \alpha ~|n_1^{(1)}-t_{11}^{(1)}| + (1-\alpha) ~|n_1^{(2)}-t_{11}^{(2)}| \nonumber\\
& +\alpha ~|n_1^{(1)}+t_{11}^{(1)}| + (1-\alpha) ~|n_1^{(2)}+t_{11}^{(2)}| \nonumber\\
& +\alpha ~|n_2^{(1)}-t_{22}^{(1)}| + (1-\alpha) ~|n_2^{(2)}-t_{22}^{(2)}| \nonumber\\
& +\alpha ~|n_2^{(1)}+t_{22}^{(1)}| + (1-\alpha) ~|n_2^{(2)}+t_{22}^{(2)}| \Big] \nonumber\\
&= \alpha I_2(\rho_{AB}^{(1)}) + (1-\alpha) I_2(\rho_{AB}^{(2)}) \nonumber\\
&\leq \sqrt{2}
\end{align}
 This completes the proof of inequality(\ref{con}) which means that any convex combination of two density matrices chosen from the set $\mathcal{S}$ lies in $\mathcal{S}$, showing that $\mathcal{S}$ is a convex set.

We now prove that the set $\mathcal{S}=\lbrace\rho_{AB}: I_2(\rho_{AB})\leq \sqrt{2}\rbrace$ is a compact set. A set can be called as compact if it is both bounded and closed. 
\begin{lemma}
$\mathcal{S}$ is a bounded set.
\end{lemma}
All elements of $\mathcal{S}$ can be characterised by upper and lower bounds which are real numbers as the eigenspectrum of any density matrix $\rho_{AB}$ in $\mathbb{C}^2 \otimes \mathbb{C}^2$ lies in $[0,1]$ and $\operatorname{Tr}[\rho_{AB}]=1$. Hence $\mathcal{S}$ is a bounded set.
\begin{lemma}
The function $I_2(\rho_{AB})$ of the density matrices $\rho_{AB}$ in $\mathbb{C}^2 \otimes \mathbb{C}^2$ is continuous, i.e. $I_2: \mathbb{C}^2 \otimes \mathbb{C}^2 \rightarrow \mathbb{R}$ is continuous at $\rho_{AB}^0 \in \mathbb{C}^2 \otimes \mathbb{C}^2$ for every $\varepsilon > 0$, there exists a $\delta > 0$ such that $\Vert\rho_{AB}-\rho_{AB}^0\Vert <\delta$ implies $|I_2(\rho_{AB})-I_2(\rho_{AB}^0)|<\varepsilon ~\forall \rho_{AB} \in \mathbb{C}^2 \otimes \mathbb{C}^2$.
\end{lemma}
Let us represent the density matrix $\rho_{AB}^0$ in terms of Eq.(\ref{gen}) with no loss of generality as
\begin{equation}
\rho_{AB}^0 = \frac{1}{4} [\openone_2 \otimes \openone_2 + \vec{m}^0.\vec{\sigma} \otimes \openone_2 + \openone_2 \otimes \vec{n}^0.\vec{\sigma} + \sum_{i,j=1}^3 t_{ij}^0 \sigma_i \otimes \sigma_j]
\label{gen0}
\end{equation}
such that $m_i-m_i^0 =\mathfrak{a}_i, n_i-n_i^0 =\mathfrak{b}_i, t_{ij}-t_{ij}^0 =\mathfrak{c}_{ij} ~\forall i,j\in \lbrace 1,2,3 \rbrace$. Now we can write that
\begin{align}
\Vert \rho_{AB}-\rho_{AB}^0\Vert &\leq \frac{c}{4} \Vert \sum_{i,j}[(\sigma_i \otimes \openone_2) + (\openone_2 \otimes \sigma_i) + (\sigma_i \otimes \sigma_j)] \Vert \nonumber\\
&= \frac{\sqrt{15}}{2} c = \delta
\end{align}
where $c=\max \lbrace |\mathfrak{a}_i|, |\mathfrak{b}_i|, |\mathfrak{c}_{ij}|\rbrace_{i,j} ~\forall i,j\in \lbrace 1,2,3 \rbrace$ and $\delta>0$.

Next, by applying Eq.(\ref{IGen}) for the density matrices $\rho_{AB}$ and $\rho_{AB}^0$, we have
\begin{align}
|I_2(\rho_{AB})-I_2(\rho_{AB}^0)| =& \frac{1}{2} \Big|(|n_1-t_{11}|+|n_1+t_{11}|+|n_2-t_{22}| \nonumber\\
&+|n_2+t_{22}|-|n_1^0 -t_{11}^0|-|n_1^0 +t_{11}^0| \nonumber\\
&-|n_2^0 -t_{22}^0|-|n_2^0 +t_{22}^0|)\Big| \nonumber\\
\leq & \frac{1}{2} \Big| |n_1-t_{11}-n_1^0 +t_{11}^0| \nonumber\\
&+ |n_1+t_{11}-n_1^0 -t_{11}^0| \nonumber\\
&+ |n_2-t_{22}-n_2^0 +t_{22}^0| \nonumber\\
&+ |n_2+t_{22}-n_2^0 -t_{22}^0| \Big| \nonumber\\
=& \frac{1}{2} \Big| |\mathfrak{b}_1 -\mathfrak{c}_{11}| + |\mathfrak{b}_1 +\mathfrak{c}_{11}| \nonumber\\
&+ |\mathfrak{b}_2 -\mathfrak{c}_{22}| + |\mathfrak{b}_2 +\mathfrak{c}_{22}| \Big| \nonumber\\
\leq & \frac{1}{2} \Big[ |\mathfrak{b}_1 -\mathfrak{c}_{11}| + |\mathfrak{b}_1 +\mathfrak{c}_{11}| \nonumber\\
&+ |\mathfrak{b}_2 -\mathfrak{c}_{22}| + |\mathfrak{b}_2 +\mathfrak{c}_{22}| \Big] \nonumber\\
\leq & |\mathfrak{b}_1| + |\mathfrak{b}_2| + |\mathfrak{c}_{11}| + |\mathfrak{c}_{22}| \nonumber\\
\leq & 4c = \frac{8}{\sqrt{15}} \delta = \varepsilon
\end{align}
where the first inequality follows from $|z_1| - |z_2| \leq |z_1 - z_2| ~\forall \lbrace z_1,z_2\rbrace \in \mathbb{R}^3$, while the second and third inequalities follow from $|\sum_i z_i| \leq \sum_i |z_i| ~\forall \lbrace z_i\rbrace_i \in \mathbb{R}^3$ and the constant, $\varepsilon>0$. This shows that $I_2(\rho_{AB})$ is a continuous function of the density matrices in $\mathbb{C}^2 \otimes \mathbb{C}^2$.
\begin{lemma}
$\mathcal{S}$ is a closed set.
\end{lemma}
The continuous function $I_2(\rho_{AB})$ is closed within $[0,2]$ for all bipartite qubit states where the lower bound occurs for maximally mixed state and the upper bound occurs for Bell states. By taking $\mathcal{S}$ under consideration, we have $I_2(\rho_{AB}) \in [0,\sqrt{2}]$. Therefore, we can write that $\mathcal{S}=\lbrace\rho_{AB}: I_2(\rho_{AB})\leq \sqrt{2}\rbrace = \lbrace I_2^{-1}([0,\sqrt{2}]) \rbrace$ which characterizes $\mathcal{S}$ as a closed set. 

As the set, $\mathcal{S}$ is both bounded and closed, it can be termed as a compact set. This completes the proof of the theorem.
\end{proof}

Now, according to the hyperplane separation theorem ~\citep{Holmes,Rudin}, if a set $\mathcal{S}$ is both convex and compact, then there exists a hyperplane between a state $\varrho_{AB}\notin \mathcal{S}$ and $\mathcal{S}=\lbrace\rho_{AB}: I_2(\rho_{AB})\leq \sqrt{2}\rbrace$. This characterizes $\mathcal{S}$ as a set of free states which are not useful as nonlocal resource for steering imaginarity at Bob's side. As the states useful for steering imaginarity can always be separated from the states not useful for the same, a witness operator may be constructed to demonstrate the steering of imaginarity in 2-measurement basis by achieving the values of $I_2(\rho_{AB})$ higher than $\sqrt{2}$. Thus a state, $\varrho_{AB}$ can be found which is witnessed by such an operator, thus allowing to  witness the imaginarity steering criterion.

\subsection{Construction of a Witness Operator}

By using Eq.(\ref{IGen}), a Hermitian operator $\widetilde{W}$ can always be constructed which obeys the conditions, (i) $\operatorname{Tr}[\widetilde{W} \rho_{AB}] \geq 0$ for the states $\rho_{AB}$ which satisfy inequality(\ref{ISI1}), and (ii) $\operatorname{Tr}[\widetilde{W} \rho_{AB}] < 0$ for at least one state $\rho_{AB}$ which violates inequality(\ref{ISI1}) and hence certifies quantum imaginarity as a nonlocal resource. The observable, $\widetilde{W}$ can be called as a witness for steering imaginarity. We propose the construction of $\widetilde{W}$ depending upon four conditions on real parametric space, as given below.
\begin{itemize}
\item[(i)] $n_1 \geq t_{11}$ \& $n_2 \geq t_{22}$,
\item[(ii)] $n_1 \geq t_{11}$ \& $n_2 \leq t_{22}$,
\item[(iii)] $n_1 \leq t_{11}$ \& $n_2 \geq t_{22}$ and 
\item[(iv)] $n_1 \leq t_{11}$ \& $n_2 \leq t_{22}$
\end{itemize} 
Eq.(\ref{IGen}) can have the following forms depending upon the aforementioned conditions as follows:
\begin{align*}
I_2(\rho_{AB}) =& \lbrace n_1, -t_{11} \rbrace + \lbrace n_2, -t_{22} \rbrace, ~~~...(\text{cond.(i)}) \\
=& \lbrace n_1, -t_{11} \rbrace + \lbrace -n_2, t_{22} \rbrace, ~~~...(\text{cond.(ii)}) \\
=& \lbrace -n_1, t_{11} \rbrace + \lbrace n_2, -t_{22} \rbrace, ~~~...(\text{cond.(iii)}) \\
=& \lbrace -n_1, t_{11} \rbrace + \lbrace -n_2, t_{22} \rbrace. ~~~...(\text{cond.(iv)}) \\
\end{align*}

In general, $I_2(\rho_{AB})$ can be expressed in terms of the witness operators as 
\begin{equation}
I_2(\rho_{AB}) = \begin{cases} \pm n_1 \pm n_2 = \operatorname{Tr}[W_{i,j}^1 ~\rho_{AB}] \\
\pm n_1 \pm t_{22} = \operatorname{Tr}[W_{i,j}^2 ~\rho_{AB}] \\
\pm t_{11} \pm n_2 = \operatorname{Tr}[W_{i,j}^3 ~\rho_{AB}] \\
\pm t_{11} \pm t_{22} = \operatorname{Tr}[W_{i,j}^4 ~\rho_{AB}]\\
\end{cases}, `\forall i,j \in \lbrace 0,1 \rbrace
\end{equation}
where $W_{i,j}^k$ ($i,j \in \lbrace 0,1 \rbrace, k \in \lbrace 1,2,3,4 \rbrace$) is a Hermitian operator yielding real expectation values and can be written as
\begin{eqnarray}
W_{i,j}^1 = (-1)^i \openone_2 \otimes \sigma_x + (-1)^j \openone_2 \otimes \sigma_y \\
W_{i,j}^2 = (-1)^i \openone_2 \otimes \sigma_x + (-1)^j \sigma_y \otimes \sigma_y \\
W_{i,j}^3 = (-1)^i \sigma_x \otimes \sigma_x + (-1)^j \openone_2 \otimes \sigma_y \\
W_{i,j}^4 = (-1)^i \sigma_x \otimes \sigma_x + (-1)^j \sigma_y \otimes \sigma_y 
\end{eqnarray}
It can be easily checked that $I_2(\rho_{AB}) \geq 0$. By applying inequality(\ref{ISI1}),  $I_2(\rho_{AB}) \leq \sqrt{2}$ implying $\operatorname{Tr}[\widetilde{W}_{i,j}^k ~\rho_{AB}] \geq 0$, where
\begin{align}
\widetilde{W}_{i,j}^k = \sqrt{2} ~\openone_4 - W_{i,j}^k, ~\forall & i,j \in \lbrace 0,1 \rbrace, \nonumber\\
& k \in \lbrace 1,2,3,4 \rbrace
\label{Witnessop}
\end{align}
 $\rho_{AB} \in \mathcal{S}$  is not useful for steering imaginarity from Alice to Bob. Though the aforesaid form of the witness operator is unnormalised, it can be normalised as $\frac{1}{4\sqrt{2}}\widetilde{W}_{i,j}^k$, so that $\operatorname{Tr}[\frac{1}{4\sqrt{2}} \widetilde{W}_{i,j}^k] = 1$. We  consider the unnormalised form of witness for the rest of the paper without loss of generality.
The violation of inequality(\ref{ISI1}), i.e., $I_2(\rho_{AB}) > \sqrt{2}$ is indicated by $\operatorname{Tr}[\widetilde{W}_{i,j}^k ~\rho_{AB}] < 0$ for at least one $\rho_{AB} \notin \mathcal{S}$, which manifests quantum steerability of imaginarity from Alice to Bob. Therefore, $\widetilde{W}_{i,j}^k$ acts as a bonafide witness operator for steering imaginarity. 

\subsection{Illustrations}

To show the usefulness of our proposed witness for steering of imaginarity, we display examples by considering the following class of bipartite qubit states.

\subsubsection*{Example 1 : Werner state} 

The single parameter bipartite Werner state~\citep{Werner} is the simplest testing ground for investigating the competence of our proposed witness to detect the steerability of quantum imaginarity from Alice to Bob by using the visibility or mixedness of the state. The general bipartite qubit state given by Eq.(\ref{gen}) can be reduced to the class of Werner states($\rho_{AB}^{\text{W}}$) when $m_1=m_2=m_3=n_1=n_2=n_3=t_{12}=t_{13}=t_{21}=t_{23}=t_{31}=t_{32}=0$ and $t_{11}=t_{22}=t_{33}=-v$ where $v$ ($0\leq v \leq 1$) is called the visibility parameter. The Werner state has the following representation,
\begin{equation}
\rho_{AB}^{\text{W}} = v ~|\psi^-\rangle\langle \psi^-| + \frac{1-v}{4} \openone_2 \otimes \openone_2
\label{Werner}
\end{equation}
where, the singlet state, $|\psi^-\rangle = \frac{1}{\sqrt{2}} (|01\rangle - |10\rangle)$. $\lbrace |0\rangle, |1\rangle \rbrace$ forms the computational basis. Using the Werner state, we find that, $p(\rho_{a_2|x}) = p(\rho_{a_1|y}) = \frac{1}{2}$ and $\mathscr{I}_R^x(\rho_{a_1|y}) = \mathscr{I}_R^y(\rho_{a_2|x}) = |v| ~\forall a_1,a_2$. Hence, the left hand side of inequality (\ref{ISI1}) becomes,
\begin{equation}
I_2(\rho_{AB}^{\text{W}}) = 2v
\end{equation}
It violates the classical upper bound $\sqrt{2}$, when $\frac{1}{\sqrt{2}} < v \leq 1$. It follows that the quantum mechanical maximum of $I_2(\rho_{AB}^{\text{W}})$ attains the value of 2 for the singlet state with $v=1$ and this is also the algebraic maximum of $I_2(\rho_{AB}^{\text{W}})$. As all the Bell states are equivalent upto local unitary transformations, hence all the maximally entangled bipartite pure states can attain the maximum violation of ISI in the 2-measurement scenario.

We can identify the parametric regime for the class of Werner states as $0=n_1\geq t_{11}$ and $0=n_2\geq t_{22}$ for all $v\in [0,1]$. Hence the appropriate Witness operator for $\rho_{AB}^{\text{W}}$ can be recognised as
\begin{equation}
\widetilde{W}_{1,1}^4 = \sqrt{2} ~\openone_4 + \sigma_x \otimes \sigma_x + \sigma_y \otimes \sigma_y
\label{WitWer}
\end{equation}
It can be easily derived that, $\operatorname{Tr}[\widetilde{W}_{1,1}^4 ~\rho_{AB}^{\text{W}}] = \sqrt{2} - 2v$, which is positive semi-definite when $v \leq \frac{1}{\sqrt{2}}$, whereas negative when $v > \frac{1}{\sqrt{2}}$ representing the Werner states which manifest the steerability of imaginarity from Alice to Bob. 

\subsubsection*{Example 2: X-state}

Two-qubit X-state has been used in the literature to study quantum correlations~\citep{Yu2007,Rau2009,Ali2010,Kelleher2021} as it encompasses a large class of entangled and separable states with just 7  real parameters. By spanning over the computational basis, $\lbrace |00\rangle, |01\rangle, |10\rangle, |11\rangle \rbrace$, the density matrix of X-state takes the following form,
\begin{equation}
\rho_{AB}^{X} = \begin{pmatrix}
\rho_{11} & 0 & 0 & \rho_{14} \\
0 & \rho_{22} & \rho_{23} & 0 \\
0 & \rho_{32} & \rho_{33} & 0 \\
\rho_{41} & 0 & 0 & \rho_{44}
\end{pmatrix}
\label{XState}
\end{equation} 
where, $\sum_{i=1}^4 \rho_{ii} = 1, \rho_{22}~\rho_{33} \geq |\rho_{23}|^2$ and $\rho_{11}~\rho_{44} \geq |\rho_{14}|^2$ in order to satisfy the positive trace class condition of the density matrix. The real diagonal and complex anti-diagonal elements of Eq.(\ref{XState}) symmetrically resembles the letter "X" and thus depends on 7 independent real state parameters which are all useful for the detection of entanglement and discord~\citep{Ali2010}. Such representation in terms of Eq.(\ref{gen}) may be re-written as,
\begin{align}
\rho_{AB}^{X} =& \frac{1}{4} [\openone_2 \otimes \openone_2 + \beta_{z0}~ \sigma_z \otimes \openone_2 + \beta_{0z}~ \openone_2 \otimes \sigma_z + \beta_{xx}~ \sigma_x \otimes \sigma_x + \nonumber\\
& \beta_{xy}~ \sigma_x \otimes \sigma_y + \beta_{yx}~ \sigma_y \otimes \sigma_x + \beta_{yy}~ \sigma_y \otimes \sigma_y + \beta_{zz}~ \sigma_z \otimes \sigma_z]
\label{GenX}
\end{align} 
where $\rho_{11}= \frac{1}{4}(1+\beta_{z0}+\beta_{0z}+\beta_{zz})$, $\rho_{22}= \frac{1}{4}(1+\beta_{z0}-\beta_{0z}-\beta_{zz})$, $\rho_{33}= \frac{1}{4}(1-\beta_{z0}+\beta_{0z}-\beta_{zz})$, $\rho_{14}= \rho_{41}^{*} = \frac{1}{4} (\beta_{xx}-\beta_{yy}) - \frac{i}{4} (\beta_{xy}+\beta_{yx})$ and $\rho_{23} = \rho_{32}^{*} = \frac{1}{4} (\beta_{xx}+\beta_{yy}) + \frac{i}{4} (\beta_{xy}-\beta_{yx})$ mark the presence of 7 real parameters in X-state~\citep{Kelleher2021}. By comparing with Eq.(\ref{gen}), we have $m_1=m_2=n_1=n_2=t_{13}=t_{31}=t_{23}=t_{32}=t_{33}=0, m_3=\beta_{z0}, n_3=\beta_{0z}, t_{11}=\beta_{xx}, t_{12}=\beta_{xy}, t_{21}=\beta_{yx}$ and $t_{22}=\beta_{yy}$.

Using the two-qubit X-state, we have, $p(\rho_{a_2|x}) = p(\rho_{a_1|y}) = \frac{1}{2}$, $\mathscr{I}_R^x(\rho_{a_1|y}) = |\beta_{yy}| ~\forall a_1$ and $\mathscr{I}_R^y(\rho_{a_2|x}) = |\beta_{xx}| ~\forall a_2$. We can evaluate the left hand side of inequality (\ref{ISI1}), given by,
\begin{equation}
I_2(\rho_{AB}^X) = |\beta_{xx}| + |\beta_{yy}|
\label{IX}
\end{equation}
which violates the unsteerability bound $\sqrt{2}$ in the green shaded region depicted in Fig.\ref{ExX}. It follows that the quantum mechanical maximum of $I_2(\rho_{AB}^X)$ attains the  algebraic maximum of 2 for $\beta_{xx}=\beta_{yy}=1$ in the 2-measurement scenario. Bell states can thus achieve optimality in terms of the quantum mechanical violation of ISI. Notably, it can be recognised that only 2 real parameters of the 7 parameter X-state are adequate to reveal the steerability through ISI.

\begin{figure}[!ht]
\includegraphics[width=6.5cm]{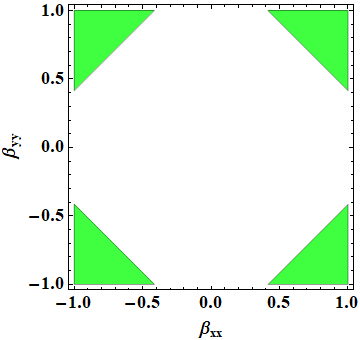}
\caption{\footnotesize (Color Online) The green coloured region in the parametric space of $\lbrace \beta_{xx},\beta_{yy} \rbrace \in [-1,1]$ indicates the violation of inequality(\ref{ISI1}).}
\label{ExX}
\end{figure}

To construct the witness operator for two-qubit X-state, we consider the following cases:
\begin{itemize}
\item[(i)] $\beta_{xx}\leq 0 ~\&~ \beta_{yy}\leq 0$ which is equivalent to $0=n_1\geq t_{11}$ and $0=n_2\geq t_{22}$ as per Eq.(\ref{gen}). Hence we obtain $I_2(\rho_{AB}^X) = -t_{11} -t_{22}$. This gives suitable witness operator as 
\begin{equation}
\widetilde{W}_{1,1}^4 = \sqrt{2} ~\openone_4 + \sigma_x \otimes \sigma_x + \sigma_y \otimes \sigma_y.
\label{WitX1}
\end{equation}
\item[(ii)] $\beta_{xx}\leq 0 ~\&~ \beta_{yy}\geq 0$ which is analogous to $0=n_1\geq t_{11}$ and $0=n_2\leq t_{22}$. By noting $I_2(\rho_{AB}^X) = -t_{11} +t_{22}$, the suitable witness operator can be given as
\begin{equation}
\widetilde{W}_{1,0}^4 = \sqrt{2} ~\openone_4 + \sigma_x \otimes \sigma_x - \sigma_y \otimes \sigma_y.
\label{WitX2}
\end{equation}
\item[(iii)] $\beta_{xx}\geq 0 ~\&~ \beta_{yy}\leq 0$ which is equivalent to $0=n_1\leq t_{11}$ and $0=n_2\geq t_{22}$. We can have $I_2(\rho_{AB}^X) = t_{11} -t_{22}$ giving rise to the suitable witness operator as 
\begin{equation}
\widetilde{W}_{0,1}^4 = \sqrt{2} ~\openone_4 - \sigma_x \otimes \sigma_x + \sigma_y \otimes \sigma_y.
\label{WitX3}
\end{equation}
\item[(iv)] $\beta_{xx}\geq 0 ~\&~ \beta_{yy}\geq 0$ or equivalently, $0=n_1\leq t_{11}$ and $0=n_2\leq t_{22}$. Hence by iterating $I_2(\rho_{AB}^X) = t_{11} +t_{22}$, the witness operator can be suitably written as
\begin{equation}
\widetilde{W}_{0,0}^4 = \sqrt{2} ~\openone_4 - \sigma_x \otimes \sigma_x - \sigma_y \otimes \sigma_y.
\label{WitX2}
\end{equation}
\end{itemize}
We can, in general, express the witness operator as $\widetilde{W}_{i,j}^4$ where $i,j\in \lbrace 0,1 \rbrace$ which can be applied to determine the value of $\operatorname{Tr}[\widetilde{W}_{i,j}^4 ~\rho_{AB}^{X}]$. The negative value is obtained in the green shaded region of the plot given by Fig.\ref{ExX}. Otherwise, it gives positive semi-definite value which is not useful for steering imaginarity from Alice to Bob under 2-measurement criterion. Therefore, our witness operator can successfully demonstrate nonlocality through the violation of inequality(\ref{ISI1}) for the two-qubit X-state. 

There exists a class of real two-qubit X-states~\citep{Munro} for which the amount of entanglement can not be increased by any unitary transformation for a given linear entropy~\citep{Bose}. These states can be called maximally entangled mixed states(MEMS). This state can be represented~\citep{Ishizaka,Vers} by its concurrence $C$ ($0\leq C \leq 1)$~\citep{Coffman} as, 
\begin{equation}
\rho_{AB}^{\text{MEMS}} = \begin{pmatrix}
\mathfrak{h}(C) & 0 & 0 & \frac{C}{2} \\
0 & 1-2\mathfrak{h}(C) & 0 & 0 \\
0 & 0 & 0 & 0 \\
\frac{C}{2} & 0 & 0 & \mathfrak{h}(C) \\
\end{pmatrix}
\label{MEMS}
\end{equation}
where $\mathfrak{h}(C) = \frac{1}{3}$ when $0 \leq C < \frac{2}{3}$ and $\mathfrak{h}(C) = \frac{C}{2}$ when $\frac{2}{3} \leq C \leq 1$.  By comparing Eq.(\ref{MEMS}) with Eq.(\ref{GenX}), we have $\beta_{z0}=-\beta_{0z}=1-2\mathfrak{h}(C), \beta_{xy}=\beta_{yx}=0, \beta_{xx}=-\beta_{yy}=C$ and $\beta_{zz}=4\mathfrak{h}(C)-1$. By applying Eq.(\ref{IX}), we obtain the left hand side of inequality(\ref{ISI1}) for MEMS as
\begin{equation}
I_2(\rho_{AB}^{\text{MEMS}}) = 2C
\end{equation}
which can optimally detect steerability for $C>\frac{1}{\sqrt{2}}$, as is 
also obtained from the  CFFW inequality given in~\citep{CFFW}. As $\beta_{xx}\geq 0$ and $\beta_{yy}\leq 0$ which corresponds to the case (iii)   above, we can  find the suitable witness operator for steering imaginarity from Alice to Bob to be
\begin{equation}
\widetilde{W}_{0,1}^4 = \sqrt{2} ~\openone_4 - \sigma_x \otimes \sigma_x + \sigma_y \otimes \sigma_y.
\label{WitMEMS}
\end{equation}
which satisfies $\operatorname{Tr}[\widetilde{W}_{0,1}^4 ~\rho_{AB}^{\text{MEMS}}]<0$ when $C>\frac{1}{\sqrt{2}}$.

\subsection{Implementation of the Witness Operator}

The demonstration of steerability of imaginarity from Alice to Bob by using the method of witness operator provides a realizable framework in lab with two measurements per side. Notably, recent works show that imaginarity of a qubit can be revealed with various applications in information processing tasks~\citep{Fernandes2024,Zhang2025}. Experimental measurability of witness operator in the two-qubit scenario may be facilitated by the decomposition of it into the sum of projectors on product subspace~\citep{Guhne2002,Mohamed2004}, i.e. $\widetilde{W} = \sum_{i=1}^n \gamma_i ~P_i \otimes P'_i$ where the co-efficients $\gamma_i$ satisfy $\sum_i \gamma_i =1$ and $P_i$, $P'_i$ are the projector in the subspaces of Alice's subsystem and Bob's subsystem, respectively. In an optimal decomposition, the number of local projectors is minimized over all possible decompositions, and the minimum number of non-vanishing $\gamma_i$ is called as the optimal number of local projectors, i.e. $\min_{\text{decomp}} n$~\citep{Guhne2002}. 

To experimentally implement the witness operators given by Eq.(\ref{Witnessop}) for two spin-$\frac{1}{2}$ particles, we have to consider projective spin measurements in $x-$ and $y-$ basis at either side of $\rho_{AB}$ respectively. Hence only two numbers of local spin-measurement settings among three mutually orthogonal unbiased basis suffice in $\mathbb{C}^2$. Thus lower number of local measurements is required to witness steering of imaginarity from Alice to Bob than that required for entanglement, quantum teleportation or nonlocal correlations (i.e. 3)~\citep{Barbieri2003,Ganguly2011,Hyllus}. 

We write the local projectors in $x-$basis as $\lbrace|0_x\rangle \equiv \frac{|0\rangle + |1\rangle}{\sqrt{2}}, |1_x\rangle \equiv \frac{|0\rangle - |1\rangle}{\sqrt{2}}\rbrace$ and in $y-$basis as $\lbrace|0_y\rangle \equiv \frac{|0\rangle + i|1\rangle}{\sqrt{2}}, |1_y\rangle \equiv \frac{|0\rangle - i|1\rangle}{\sqrt{2}}\rbrace$ respectively where  $\lbrace |0\rangle, |1\rangle \rbrace$ form the projectors in $z-$basis. By using the completeness relation, $|0\rangle\langle 0| + |1\rangle\langle 1| = |0_x\rangle\langle 0_x| + |1_x\rangle\langle 1_x| = |0_y\rangle\langle 0_y| + |1_y\rangle\langle 1_y|=\openone_2$, Eq.(\ref{Witnessop}) can be recast as
\begin{widetext}
\begin{align}
\widetilde{W}_{i,j}^1 =& \nu_1 (|0_x 0_x\rangle\langle 0_x 0_x| + |1_x 0_x\rangle\langle 1_x 0_x|) + \nu_2 (|0_x 1_x\rangle\langle 0_x 1_x| + |1_x 1_x\rangle\langle 1_x 1_x|) \nonumber\\
&+ \nu_3 (|0_y 0_y\rangle\langle 0_y 0_y| + |1_y 0_y\rangle\langle 1_y 0_y|) + \nu_4 (|0_y 1_y\rangle\langle 0_y 1_y| + |1_y 1_y\rangle\langle 1_y 1_y|), \\
\widetilde{W}_{i,j}^2 =& \nu_1 (|0_x 0_x\rangle\langle 0_x 0_x| + |1_x 0_x\rangle\langle 1_x 0_x|) + \nu_2 (|0_x 1_x\rangle\langle 0_x 1_x| + |1_x 1_x\rangle\langle 1_x 1_x|) \nonumber\\
&+ \nu_3 (|0_y 0_y\rangle\langle 0_y 0_y| + |1_y 1_y\rangle\langle 1_y 1_y|) + \nu_4 (|1_y 0_y\rangle\langle 1_y 0_y| + |0_y 1_y\rangle\langle 0_y 1_y|),\\
\widetilde{W}_{i,j}^3 =& \nu_1 (|0_x 0_x\rangle\langle 0_x 0_x| + |1_x 1_x \rangle\langle 1_x 1_x|) + \nu_2 (|1_x 0_x\rangle\langle 1_x 0_x| + |0_x 1_x\rangle\langle 0_x 1_x|) \nonumber\\
&+ \nu_3 (|0_y 0_y\rangle\langle 0_y 0_y| + |1_y 0_y\rangle\langle 1_y 0_y|) + \nu_4 (|0_y 1_y\rangle\langle 0_y 1_y| + |1_y 1_y\rangle\langle 1_y 1_y|), \\
\widetilde{W}_{i,j}^4 =& \nu_1 (|0_x 0_x\rangle\langle 0_x 0_x| + |1_x 1_x \rangle\langle 1_x 1_x|) + \nu_2 (|1_x 0_x\rangle\langle 1_x 0_x| + |0_x 1_x\rangle\langle 0_x 1_x|) \nonumber\\
&+ \nu_3 (|0_y 0_y\rangle\langle 0_y 0_y| + |1_y 1_y\rangle\langle 1_y 1_y|) + \nu_4 (|1_y 0_y\rangle\langle 1_y 0_y| + |0_y 1_y\rangle\langle 0_y 1_y|),
\end{align}
\end{widetext}
where, $\nu_1=\sqrt{2}+(-1)^{i+1}$, $\nu_2=\sqrt{2}+(-1)^{i}$, $\nu_3=(-1)^{j+1}$ and $\nu_4=(-1)^{j}$ are non-vanishing co-efficients. These witness operators have asymmetric representation. In a given experimental scenario, the sign of the expectation value of $\widetilde{W}_{i,j}^k$ ($i,j\in\lbrace 0,1\rbrace, k\in\lbrace 1,2,3,4 \rbrace$), i.e. $\langle \widetilde{W}_{i,j}^k \rangle_{\rho_{AB}}$ determines the presence of nonlocality in a bipartite correlation in terms of the steering of imaginarity. Here the optimal number of local projectors is $\min_{\text{decomp}} n = 8$ which is lower than the number of copies of bipartite qubit states required to reveal the presence of quantum mechanical correlation or the usefulness of it for performing an information processing task with non-classical advantage. The lower bound of the expectation value of the witness operators may be adjusted to a positive constant to demonstrate ISI under inevitable mixing of noise in practical scenario. Thus minimization of error to estimate the expectation value of the witness operator is necessary to tolerate such noise.

\section{Monogamy of Steering of Imaginarity} \label{C5}

Most quantum mechanical correlations follow monogamy criterion under the no-signalling framework when  more than two parties who share such correlation are involved. A monogamy relation for entanglement among three qubits was first demonstrated in~\citep{CKW2000}. Further studies on monogamy relations in the context of Bell nonlocality~\citep{Toner2006,Toner2008,Kurz2011}, EPR-steering~\citep{Reid2013,Milne2014}, quantum teleportation~\citep{Lee2009}, contextual inequalities~\citep{Ramanathan2012,Kurz2014} have been performed and their importance in several information-theoretic applications have consequently been shown~\citep{Giorgi2011,Allegra2011,Song2013,Kumar2016,Dhar2017}. Monogamy has significant implication in quantum cryptographic networks~\citep{Gisin2002,TP2014,Datta2017}.

Here we consider the general tripartite pure state to derive a monogamy relation for our imaginarity steering inequality in the 2-measurement set-up.

\begin{theorem}
If Alice(A), Bob(B) and Charlie(C) share a tripartite state $|\psi\rangle_{ABC}$, then the left hand side of inequality(\ref{ISI1}) corresponding to Alice$\rightarrow$Bob, i.e. $I_2(\rho_{AB})$ and that corresponding to Alice$\rightarrow$Charlie, i.e. $I_2(\rho_{AC})$ satisfy the following relation,
\begin{equation}
I_2(\rho_{AB}) + I_2(\rho_{AC}) \leq 2\sqrt{2}
\end{equation}
where $\rho_{AB}=\operatorname{Tr}_{C}[|\psi\rangle_{ABC}\langle\psi|]$ and $\rho_{AC}=\operatorname{Tr}_{B}[|\psi\rangle_{ABC}\langle\psi|]$.
\end{theorem}

\begin{proof}
An arbitrary pure state in $\mathbb{C}^2 \otimes \mathbb{C}^2 \otimes \mathbb{C}^2$ can be considered as the union of the well-known GHZ state~\citep{GHZ1990} and the W state~\citep{Dur2000}. Up to local operations, such a state can be generalised in the computational basis ($\lbrace |0\rangle,|1\rangle \rbrace$) as~\citep{Acin2000,Acin2001}
\begin{equation}
|\psi\rangle_{ABC} = \eta_0 |000\rangle + \eta_1 \exp(i \theta) |100\rangle + \eta_2 |101\rangle + \eta_3 |110\rangle + \eta_4 |111\rangle,
\label{gentri}
\end{equation}
where $\lbrace \eta_i \rbrace_{i=0}^{4},\theta \in \mathbb{R}^3,  \sum_{i=0}^4 \eta_i^2 =1, 0\leq \eta_i \leq 1$ and $0\leq \theta \leq \pi$. The GHZ-state can be obtained from Eq.(\ref{gentri}) by putting $\eta_0 = \eta_4 = \frac{1}{\sqrt{2}}, \eta_1 = \eta_2 = \eta_3 = 0$ and the W-state  by putting $\eta_4 = \theta = 0$.

The reduced state between Alice and Bob from the tripartite joint state given by Eq.(\ref{gentri}) can be obtained by taking partial trace over Charlie's subsystem as $\rho_{AB}=\operatorname{Tr}_{C}[|\psi\rangle_{ABC}\langle\psi|]$ and the reduced state between Alice and Charlie can be expressed in a similar way as $\rho_{AC}=\operatorname{Tr}_{B}[|\psi\rangle_{ABC}\langle\psi|]$, \textit{viz.}
\begin{eqnarray}
\rho_{AB}= \begin{pmatrix}
\eta_0^2 & 0 & e^{-i\theta} \eta_0 \eta_1 & \eta_0 \eta_3 \\
0 & 0 & 0 & 0 \\
e^{i\theta} \eta_0 \eta_1 & 0 & \eta_1^2 + \eta_2^2 & e^{i\theta} \eta_1 \eta_3 + \eta_2 \eta_4 \\
\eta_0 \eta_3 & 0 & e^{-i\theta} \eta_1 \eta_3 + \eta_2 \eta_4 & 1 -\eta_0^2 - \eta_1^2 - \eta_2^2
\end{pmatrix}, \\
\rho_{AC}= \begin{pmatrix}
\eta_0^2 & 0 & e^{-i\theta} \eta_0 \eta_1 & \eta_0 \eta_2 \\
0 & 0 & 0 & 0 \\
e^{i\theta} \eta_0 \eta_1 & 0 & \eta_1^2 + \eta_3^2 & e^{i\theta} \eta_1 \eta_2 + \eta_3 \eta_4 \\
\eta_0 \eta_2 & 0 & e^{-i\theta} \eta_1 \eta_2 + \eta_3 \eta_4 & 1 -\eta_0^2 - \eta_1^2 - \eta_3^2
\end{pmatrix}.
\end{eqnarray} 
Let us suppose that Alice wants to steer the imaginarity of both the subsystems possessed by Bob and Charlie in a preferred basis of $\sigma_x$ and $\sigma_y$, by doing the measurements on her subsystem in $y-$ and $x-$ directions, respectively,  by  employing the copies of $\rho_{AB}$ and $\rho_{AC}$ separately. For both $\rho_{AB}$ and $\rho_{AC}$, we have the probabilities, $p(\pm|x)= \frac{1}{2} \pm \eta_0 \eta_1 \cos \theta$ and $p(\pm|y) = \frac{1}{2} \pm \eta_0 \eta_1 \sin \theta$ by making measurements in x- and y- directions respectively, at Alice's side along with $\pm$ as outputs. 

If Bob computes imaginarity of his local state in the y-basis depending upon the measurement done at Alice's side in the x-direction by on the state $\rho_{AB}$, then 
\begin{equation}
\mathscr{I}_R^y(\rho^{B}_{\pm|x}) = 2\big|\frac{\eta_2 \eta_4 + \eta_3(\eta_1 \cos \theta \pm \eta_0)}{1 \pm 2 \eta_0 \eta_1 \cos \theta}\big|
\label{ItriABx}
\end{equation}
with respect to the outcomes $\pm$ at Alice's side.
Similarly, the imaginarity of Bob's subsystem in x-basis after Alice's measurement in y-direction with the use of $\rho_{AB}$ can be written as
\begin{equation}
\mathscr{I}_R^x(\rho^{B}_{\pm|y}) = 2\big|\frac{\eta_3(\eta_1 \sin \theta \pm \eta_0)}{1 \pm 2 \eta_0 \eta_1 \sin \theta}\big|
\label{ItriABy}
\end{equation}
with respect to the outcomes $\pm$ at Alice's side.

In a similar fashion, Charlie computes imaginarity in $\lbrace y,x \rbrace$-basis contingent upon Alice's measurement settings in the direction of $\lbrace \sigma_x, \sigma_y \rbrace$ respectively, pertaining to $\rho_{AC}$. Hence, we get 
\begin{equation}
\mathscr{I}_R^y(\rho^{C}_{\pm|x}) = 2\big|\frac{\eta_3 \eta_4 + \eta_2(\eta_1 \cos \theta \pm \eta_0)}{1 \pm 2 \eta_0 \eta_1 \cos \theta}\big|
\label{ItriACx}
\end{equation}
and
\begin{equation}
\mathscr{I}_R^x(\rho^{C}_{\pm|y}) = 2\big|\frac{\eta_2(\eta_1 \sin \theta \pm \eta_0)}{1 \pm 2 \eta_0 \eta_1 \sin \theta}\big|
\label{ItriACy}
\end{equation}
analogous to the results $\pm$ obtained by Alice.

Now we make use of Eqs.(\ref{ItriABx}-\ref{ItriACy}) in inequalities(\ref{ISI1}) to obtain
\begin{align}
&I_2(\rho_{AB}) + I_2(\rho_{AC}) \nonumber\\
&= \Big(\sum_{a_1=+,-} p(a_1|y) ~\mathscr{I}_R^x(\rho^B_{a_1|y}) + \sum_{a_2=+,-} p(a_2|x) ~\mathscr{I}_R^y(\rho^B_{a_2|x})\Big) \nonumber\\
&+ \Big(\sum_{a_1=+,-} p(a_1|y) ~\mathscr{I}_R^x(\rho^C_{a_1|y}) + \sum_{a_2=+,-} p(a_2|x) ~\mathscr{I}_R^y(\rho^C_{a_2|x})\Big) \nonumber\\
&\leq 2\sqrt{2}
\label{monotri}
\end{align}
where the numerical maximum $2\sqrt{2}$ is obtained for $\eta_0=\frac{1}{\sqrt{2}}, \eta_1=\eta_4=0, \eta_2=\eta_3=\frac{1}{2}$ ($\theta$ arbitrary). We can also check that, $I_2(\rho_{AB}) = \sum_{a_1=+,-} p(a_1|y) ~\mathscr{I}_R^x(\rho^B_{a_1|y}) + \sum_{a_2=+,-} p(a_2|x) ~\mathscr{I}_R^y(\rho^B_{a_2|x}) \leq 2$ by using $\rho_{AB}$ where the upper bound corresponds to $\eta_0=\eta_3=\frac{1}{\sqrt{2}}$ and $\eta_1=\eta_2=\eta_4=0$. On the other hand,  $I_2(\rho_{AC}) = \sum_{a_1=+,-} p(a_1|y) ~\mathscr{I}_R^x(\rho^C_{a_1|y}) + \sum_{a_2=+,-} p(a_2|x) ~\mathscr{I}_R^y(\rho^C_{a_2|x}) \leq 2$ by utilizing $\rho_{AC}$ where the upper bound is achieved for $\eta_0=\eta_2=\frac{1}{\sqrt{2}}$ and $\eta_1=\eta_3=\eta_4=0$. This implies that Alice can maximally steer the imaginarity of either the subsystems of Bob or Charlie with quantum advantage when the rest remains uncorrelated or classically correlated with Alice. From inequality(\ref{monotri}) it is evident that, $\rho_{AC}$ must satisfy inequality(\ref{ISI1}) when $\rho_{AB}$ violates it or vice-versa.  It means that $\rho_{AB}$ and $\rho_{AC}$ can not simultaneously have quantum advantage of the correlation through the violation of inequality(\ref{ISI1}). Or in other words, when Alice steers the imaginarity of local state possessed by Bob, then she can not steer that of Charlie in the tripartite scenario. 
\end{proof}

\section{Efficacy of imaginarity steering inequality} \label{C6}

The imaginarity steering criterion (ISI) given by inequality (\ref{ISI1}) is a 2-measurement criterion which demonstrates steerable correlations based on partial knowledge of Bob's local hidden state (LHS). There exist other steering criteria utilizing partial information about Bob's LHS, such as the nonlocal advantage of quantum coherence (NAQC)~\citep{Mondal}, and nonlocal advantage of quantum imaginarity (NAQI)~\citep{Wei2024} which define quantum steerable correlations using  three measurements per side. Our imagniarity steering criterion is distinct from both NAQC and NAQI under the steering framework. In this section, we compare our criterion with the aforesaid steering criteria to show its robustness over the others.

Coherence of a quantum state provides  partial information about the density matrix through only the off-diagonal terms of it. Quantum coherence is a basis dependent property of quantum state and can quantified by various  measures such as, $l_1$-norm ($C^{l_1}$), relative entropy ($C^{R}$) and skew information ($C^{S}$)~\citep{Baumgratz,Girolami,Winter}. For a $2\times 2$ density matrix $\varrho$, $C^{l_1}(\varrho) = \sum_{\substack{i,j \\ i \neq j}} |\varrho_{ij}|$, $C^{R}(\varrho) = S(\varrho_d) - S(\varrho)$, where $S(\varrho)$ is the Von-Neumann entropy of $\varrho$~\citep{Neumann} and the diagonal matrix $\varrho_d$ is formed by the diagonal elements of $\varrho$ in a particular basis and $C^{S}(\varrho) = \operatorname{Tr}[\varrho.\sigma_i.\sigma_i - \sqrt{\varrho}.\sigma_i.\sqrt{\varrho}.\sigma_i]$ in the basis representation of $\sigma_i$. Using the above quantifiers, complementarity relations among three mutually orthogonal basis in 2-dimensions can be constructed as $\sum_{j=x,y,z} C^{g}_j (\varrho) \leq \gamma^g$, where the upper limit $\gamma^g=\lbrace \sqrt{6}, 2.23, 2 \rbrace$ corresponds to $g = \lbrace l_1, R, S \rbrace$, respectively~\citep{Mondal}.

Depending upon the above complementarity relations, sufficient criteria for steering of quantum coherence in a scenario of 3-dichotomic measurements per side can be formulated as~\citep{Mondal},
\begin{equation}
N_3^g = \frac{1}{2} \sum_{i,j,a} p(a|j\neq i) ~C_i^{g} (\rho_{B|\Pi_{a|j\neq i}}) \leq \gamma^g
\label{NAQC}
\end{equation}
where, $i,j = \lbrace x, y, z \rbrace$, $a = \lbrace 0,1 \rbrace$, $p(a|j\neq i) = \operatorname{Tr}[(\Pi_{a|j\neq i} \otimes \openone_2) \rho_{AB}]$ for bipartite state $\rho_{AB}$ between Alice and Bob and $\rho_{B|\Pi_{a|j\neq i}} = \frac{1}{p(a|j\neq i)} \operatorname{Tr}_A [(\Pi_{a|j\neq i} \otimes \openone_2) \rho_{AB}]$. The violation of inequalities (\ref{NAQC}) implies the steering of local quantum coherence of Bob by Alice and thus, it gives rise to nonlocal advantage of quantum coherence(NAQC).

Another 3-measure steering inequality has been proposed recently by using $l_1$-norm and relative entropy of imaginarity by Wei et al~\citep{Wei2024}, which is given by
\begin{equation}
\overline{N}_3^{g'} = \max_{M,\Pi} \sum_{i,a} p(a|i) ~\mathscr{I}_{M_i}^{g'}(\rho_{B|\Pi_{a|i}}) \leq \gamma^{g'}
\label{NAQI}
\end{equation}
where, $\Pi$ is the set of measurements done at Alice's side, $M$ is the set of Maximally Unbiased Basis (MUBs) chosen at Bob's side, $\mathscr{I}_{M_i}^{g'}(\rho_{B|\Pi_{a|i}})$ is the $g'$-quantifier of imaginarity of Bob's conditional state depending upon Alice's measurement in $i-$direction producing outcome $a$ and $\gamma^{g'}$ is the upper bound corresponding to the $g'$-quantifier. The value of $\gamma^{g'}=\lbrace \sqrt{5}, 2.02685 \rbrace$ corresponding to $g'=\lbrace l_1, R \rbrace$, i.e., the $l_1$-norm and relative entropy measures of imaginarity, respectively. The $l_1$-norm of imaginarity for a qubit $\rho$ is expressed as $\mathscr{I}^{l_1}(\rho)= \sum_{a,b\neq a} |\text{Im}(\rho_{a,b})|$ depending upon the off-diagonal terms $\rho_{a,b} = \langle a|\rho|b \rangle$ in a fixed basis~\citep{Xue2021}. The relative entropy of imaginarity for $\rho$ is given by, $\mathscr{I}^{R}(\rho)= \operatorname{Tr}[\rho \log_2 \rho] - \operatorname{Tr}[\rho' \log_2 \rho']$ where $\rho'=\frac{\rho+\rho^T}{2}$ with $\rho^T$ as the transpose of $\rho$ in a fixed basis~\citep{Xue2021,Chen2023}. The violation of inequalities (\ref{NAQI}) implies the nonlocal advantage of quantum imaginarity(NAQI). Like NAQC, the $l_1$-norm measure ascertains the optimality of the violation.

\subsection{Robustness under mixing with white noise}

The Bell states are the maximally resourceful bipartite qubit states in terms all quantum correlation~\citep{NC2010}. If one of such states, say the singlet state, is mixed with white noise with a ratio $v:(1-v)$, then there exists a non-classical domain, upto which white noise can be tolerated to preserve the quantum correlation. The well-known Werner class of states are described by Eq.(\ref{Werner})~\citep{Werner}. If we compare the maximal tolerable range of $v$ among ISI, NAQC~\citep{Mondal} and NAQI~\citep{Wei2024}, we can observe from Table \ref{tab1} that, 
\begin{table}[h!]
\centering
 \begin{tabular}{| c | c | c |} 
 \hline
 \multicolumn{3}{|c|}{Range of v for the violation of} \\
 \hline
 ISI & NAQC ineq. & NAQI ineq. \\
 \hline
 $v> 0.707$ & $v> 0.815$ & $v> 0.745$ \\
 \hline
 \end{tabular}
\caption{\footnotesize ISI can detect the steerability of Werner state in a larger parametric domain than NAQC and NAQI. The same bounds are obtained
if Alice performs unsharp measurements with parameter $v$.}
\label{tab1}
\end{table}
ISI captures nonlocality of Werner states by tolerating the maximum amount of noise,  as also manifested identically by Bell nonlocality~\citep{Bell,CHSH,Brunner2014} and quantum steering~\citep{CFFW,Cavalcanti2017}. Hence there exists a range of $v$ where both NAQC and NAQI fail to detect steerability with the aid of partial knowledge of Bob's LHS unlike ISI given by Eq.(\ref{ISI1}) proposed by us. This shows the efficacy of ISI over NAQC and NAQI under mixing with white noise. 

\subsection{Robustness under unsharp measurements}

If Alice performs unsharp measurement on her side instead of projective measurements, then the measurement correlation pertaining to a given bipartite qubit state also changes. Unsharp measurement is a single parameter Positive Operator Valued Measurement (POVM), which maps a state onto mixed subspace of eigenvectors~\citep{Busch,Mal1}. For a qubit, it is characterized by effect operators as given by
\begin{equation}
E_{a|A} = \lambda ~\Pi_{a|A} + (1-\lambda) \frac{\openone_2}{2}
\end{equation}
corresponding to projective measurement $\Pi_{a|A}$ ($a\in\lbrace+,-\rbrace$ where $\lambda$ ($0\leq \lambda \leq 1$) is called the sharpness of measurement. The effect operators follow the completeness relation as $\sum_{a=+,-} E_{a|A}=\openone_2$. 

Depending upon Alice's unsharp measurement on her side of the joint state $\rho_{AB}$, the post-measurement state transforms as
\begin{equation}
\rho_{AB} \rightarrow \rho'_{AB} = \frac{(\sqrt{E_{a|A}}\otimes \openone_2)~\rho_{AB}~(\sqrt{E_{a|A}}\otimes \openone_2)}{p(a|A)}
\end{equation}
where the probability of getting outcomes $\pm$ from Alice's unsharp measurement is given by $p(a|A)=\operatorname{Tr}[(E_{a|A}\otimes \openone_2) \rho_{AB}]$. Hence the normalised conditional state at Bob's side becomes $\rho_{B|E_{a|A}} = \operatorname{Tr}_A[\rho'_{AB}]= \frac{\operatorname{Tr}_A[(\sqrt{E_{a|A}}\otimes \openone_2)~\rho_{AB}~(\sqrt{E_{a|A}}\otimes \openone_2)]}{\operatorname{Tr}[(\sqrt{E_{a|A}}\otimes \openone_2)~\rho_{AB}~(\sqrt{E_{a|A}}\otimes \openone_2)]}$. By using the above method for the initially chosen singlet state $|\psi^-\rangle = \frac{1}{\sqrt{2}}(|01\rangle - |10\rangle)$, we have $p(\pm|x)=p(\pm|y)=\frac{1}{2}$ and $\mathscr{I}_R^x(\rho_{B|E_{\pm|y}}) = \mathscr{I}_R^y(\rho_{B|E_{\pm|x}}) = \lambda$ for a given set of measurements. The left hand side of inequality(\ref{ISI1}) thus turns into
\begin{equation}
I_2(|\psi^-\rangle) = 2\lambda
\end{equation}
which exceeds $\sqrt{2}$ when $\lambda > \frac{1}{\sqrt{2}}$. The violation is maximal when there is no unsharpness in the measurement. 

We can observe that the conditional state, $\rho_{B|E_{a|A}}$ obtained at Bob's side after Alice's unsharp measurement on her part of the singlet state is equivalent to that derived using projective measurement done by Alice on her subsystem of the initially chosen Werner state with visibility $\lambda$, i.e., $\lambda |\psi^-\rangle\langle \psi^-| + \frac{1-\lambda}{4} \openone_2 \otimes \openone_2$. It hence follows that  the optimum range of sharpness parameter for which the quantum mechanical violation of ISI,  NAQC and NAQI occur respectively,  are indeed the same as those displayed in the Table \ref{tab1} (replacing $v$ by $\lambda$). Therefore, ISI turns out to be the most resilient against unsharp measurements done by Alice to obtain steerability  using incomplete  knowledge of Bob's local hidden state (LHS). 

\section{Conclusions}\label{C7}

Complex numbers are an essential part of the density matrix formalism in quantum mechanics. The role of imaginarity turns out to be evident in the manifestation of various features of quantum theory. The operational advantage of imaginary numbers in local state discrimination is well established~\citep{Wu1}. Since quantum mechanics allows entanglement and nonlocal features among two or more spatially separated parties, which distinguishes it from classical theory, it is imperative to investigate the role of imaginarity for displaying nonlocal features of entangled states. In the present work, we show how imaginarity forms an important resource in the efficient manifestation of quantum steerability.

In our present analysis we first construct 2-setting complementarity relations based on the basis dependent nature of imaginarity. We then present our imaginarity steering inequality (ISI) for bipartite qubit states which forms a sufficient criterion for detecting quantum steering. Next, we show that the set of all bipartite qubit states which satisfy ISI is both convex and compact. This enables formulation of witness operators which produce negative expectation value for the states  demonstrating steerability through the violation of ISI. We furnish examples of such witnesses for some well-known quantum states, and elaborate on the practical implementation of such witness operators. Additionally, we demonstrate the monogamy of the quantum correlations embodied by the violation of our imaginarity steering inequality. 

Before concluding, it may be worth re-emphasizing that our imaginarity steering inequality is a function of just four real parameters in the space of general bipartite qubit states. The framework of ISI thus emerges as an effective test-bed for detecting nonlocality with less complex empirical set-ups compared to those required for detecting Bell-nonlocality or general steerability. Since the imaginarity steering criterion proposed by us does not invoke full knowledge of the density matrix, general quantum steerability~\citep{Wiseman,Jones} may not always imply steerability of imaginarity. However, the reverse holds true. In this work our proposed ISI is further compared with two other steering criteria  based on the partial knowledge of the quantum state. We show that ISI outperforms such criteria, known as nonlocal advantage of quantum coherence~\citep{Mondal} and nonlocal advantage of quantum imaginarity~\citep{Wei2024} in the context of two different scenarios demonstrating the comparative robustness of ISI against  white noise and unsharp measurements respectively. 

The analysis present here may open up several interesting directions for further study. The possibility of quantification of steerability through witness operators~\citep{Brandao2005, Rau2009} could be investigated employing practically realizable fault tolerant techniques~\citep{Cai2023}.  The analysis of various information processing tasks such as one-sided device independent quantum key distribution~\citep{Pramanik} and random number generation~\citep{Passaro} may be simplified using our proposed framework. Since, all maximally entangled states violate the ISI maximally, it may be used for one-sided device-independent certification of maximally entangled states~\citep{Supic,Goswami} more efficiently. Moreover, the effects of steering imaginarity in accelerated frames~\citep{Mondal1} and in higher dimensional set-ups~\citep{Hu} may be a promising  avenue to explore. Further, the sharing of the imaginarity steering correlation among a number of observers~\citep{Mal1,Sasmal2018,Datta2018,Gupta2021} may be interesting to explore in future works.

\bibliography{Img_aug25_1}

\begin{thebibliography}{107}%
\makeatletter
\providecommand \@ifxundefined [1]{%
 \@ifx{#1\undefined}
}%
\providecommand \@ifnum [1]{%
 \ifnum #1\expandafter \@firstoftwo
 \else \expandafter \@secondoftwo
 \fi
}%
\providecommand \@ifx [1]{%
 \ifx #1\expandafter \@firstoftwo
 \else \expandafter \@secondoftwo
 \fi
}%
\providecommand \natexlab [1]{#1}%
\providecommand \enquote  [1]{``#1''}%
\providecommand \bibnamefont  [1]{#1}%
\providecommand \bibfnamefont [1]{#1}%
\providecommand \citenamefont [1]{#1}%
\providecommand \href@noop [0]{\@secondoftwo}%
\providecommand \href [0]{\begingroup \@sanitize@url \@href}%
\providecommand \@href[1]{\@@startlink{#1}\@@href}%
\providecommand \@@href[1]{\endgroup#1\@@endlink}%
\providecommand \@sanitize@url [0]{\catcode `\\12\catcode `\$12\catcode
  `\&12\catcode `\#12\catcode `\^12\catcode `\_12\catcode `\%12\relax}%
\providecommand \@@startlink[1]{}%
\providecommand \@@endlink[0]{}%
\providecommand \url  [0]{\begingroup\@sanitize@url \@url }%
\providecommand \@url [1]{\endgroup\@href {#1}{\urlprefix }}%
\providecommand \urlprefix  [0]{URL }%
\providecommand \Eprint [0]{\href }%
\providecommand \doibase [0]{https://doi.org/}%
\providecommand \selectlanguage [0]{\@gobble}%
\providecommand \bibinfo  [0]{\@secondoftwo}%
\providecommand \bibfield  [0]{\@secondoftwo}%
\providecommand \translation [1]{[#1]}%
\providecommand \BibitemOpen [0]{}%
\providecommand \bibitemStop [0]{}%
\providecommand \bibitemNoStop [0]{.\EOS\space}%
\providecommand \EOS [0]{\spacefactor3000\relax}%
\providecommand \BibitemShut  [1]{\csname bibitem#1\endcsname}%
\let\auto@bib@innerbib\@empty
\bibitem [{\citenamefont {Batle}\ \emph {et~al.}(2002)\citenamefont {Batle},
  \citenamefont {Plastino}, \citenamefont {Casas},\ and\ \citenamefont
  {Plastino}}]{Batle}%
  \BibitemOpen
  \bibfield  {author} {\bibinfo {author} {\bibfnamefont {J.}~\bibnamefont
  {Batle}}, \bibinfo {author} {\bibfnamefont {A.}~\bibnamefont {Plastino}},
  \bibinfo {author} {\bibfnamefont {M.}~\bibnamefont {Casas}},\ and\ \bibinfo
  {author} {\bibfnamefont {A.}~\bibnamefont {Plastino}},\ }\bibfield  {title}
  {\bibinfo {title} {On the entanglement properties of two-rebits systems},\
  }\href {https://doi.org/https://doi.org/10.1016/S0375-9601(02)00582-0}
  {\bibfield  {journal} {\bibinfo  {journal} {Phys. Lett. A}\ }\textbf
  {\bibinfo {volume} {298}},\ \bibinfo {pages} {301} (\bibinfo {year}
  {2002})}\BibitemShut {NoStop}%
\bibitem [{\citenamefont {Wootters}(2012)}]{Wootters1}%
  \BibitemOpen
  \bibfield  {author} {\bibinfo {author} {\bibfnamefont {W.~K.}\ \bibnamefont
  {Wootters}},\ }\bibfield  {title} {\bibinfo {title} {Entanglement sharing in
  real-vector-space quantum theory},\ }\href
  {https://doi.org/10.1007/s10701-010-9488-1} {\bibfield  {journal} {\bibinfo
  {journal} {Found. Phys.}\ }\textbf {\bibinfo {volume} {42}},\ \bibinfo
  {pages} {19} (\bibinfo {year} {2012})}\BibitemShut {NoStop}%
\bibitem [{\citenamefont {Wootters}(2014)}]{Wootters2}%
  \BibitemOpen
  \bibfield  {author} {\bibinfo {author} {\bibfnamefont {W.~K.}\ \bibnamefont
  {Wootters}},\ }\bibfield  {title} {\bibinfo {title} {The rebit three-tangle
  and its relation to two-qubit entanglement},\ }\href
  {https://doi.org/10.1088/1751-8113/47/42/424037} {\bibfield  {journal}
  {\bibinfo  {journal} {Jour. Phys. A: Math. Theo.}\ }\textbf {\bibinfo
  {volume} {47}},\ \bibinfo {pages} {424037} (\bibinfo {year}
  {2014})}\BibitemShut {NoStop}%
\bibitem [{\citenamefont {Prasannan}\ \emph {et~al.}(2021)\citenamefont
  {Prasannan}, \citenamefont {De}, \citenamefont {Barkhofen}, \citenamefont
  {Brecht}, \citenamefont {Silberhorn},\ and\ \citenamefont
  {Sperling}}]{Prasannan}%
  \BibitemOpen
  \bibfield  {author} {\bibinfo {author} {\bibfnamefont {N.}~\bibnamefont
  {Prasannan}}, \bibinfo {author} {\bibfnamefont {S.}~\bibnamefont {De}},
  \bibinfo {author} {\bibfnamefont {S.}~\bibnamefont {Barkhofen}}, \bibinfo
  {author} {\bibfnamefont {B.}~\bibnamefont {Brecht}}, \bibinfo {author}
  {\bibfnamefont {C.}~\bibnamefont {Silberhorn}},\ and\ \bibinfo {author}
  {\bibfnamefont {J.}~\bibnamefont {Sperling}},\ }\bibfield  {title} {\bibinfo
  {title} {Experimental entanglement characterization of two-rebit states},\
  }\href {https://doi.org/10.1103/PhysRevA.103.L040402} {\bibfield  {journal}
  {\bibinfo  {journal} {Phys. Rev. A}\ }\textbf {\bibinfo {volume} {103}},\
  \bibinfo {pages} {L040402} (\bibinfo {year} {2021})}\BibitemShut {NoStop}%
\bibitem [{\citenamefont {McKague}\ \emph {et~al.}(2009)\citenamefont
  {McKague}, \citenamefont {Mosca},\ and\ \citenamefont {Gisin}}]{McKague}%
  \BibitemOpen
  \bibfield  {author} {\bibinfo {author} {\bibfnamefont {M.}~\bibnamefont
  {McKague}}, \bibinfo {author} {\bibfnamefont {M.}~\bibnamefont {Mosca}},\
  and\ \bibinfo {author} {\bibfnamefont {N.}~\bibnamefont {Gisin}},\ }\bibfield
   {title} {\bibinfo {title} {Simulating quantum systems using real hilbert
  spaces},\ }\href {https://doi.org/10.1103/PhysRevLett.102.020505} {\bibfield
  {journal} {\bibinfo  {journal} {Phys. Rev. Lett.}\ }\textbf {\bibinfo
  {volume} {102}},\ \bibinfo {pages} {020505} (\bibinfo {year}
  {2009})}\BibitemShut {NoStop}%
\bibitem [{\citenamefont {Koh}\ \emph {et~al.}(2018)\citenamefont {Koh},
  \citenamefont {Niu},\ and\ \citenamefont {Yoder}}]{Koh}%
  \BibitemOpen
  \bibfield  {author} {\bibinfo {author} {\bibfnamefont {D.~E.}\ \bibnamefont
  {Koh}}, \bibinfo {author} {\bibfnamefont {M.~Y.}\ \bibnamefont {Niu}},\ and\
  \bibinfo {author} {\bibfnamefont {T.~J.}\ \bibnamefont {Yoder}},\ }\bibfield
  {title} {\bibinfo {title} {Quantum simulation from the bottom up: the case of
  rebits},\ }\href {https://doi.org/10.1088/1751-8121/aab9c4} {\bibfield
  {journal} {\bibinfo  {journal} {Jour. Phys. A: Math. and Theo.}\ }\textbf
  {\bibinfo {volume} {51}},\ \bibinfo {pages} {195302} (\bibinfo {year}
  {2018})}\BibitemShut {NoStop}%
\bibitem [{\citenamefont {Araki}(1980)}]{Araki}%
  \BibitemOpen
  \bibfield  {author} {\bibinfo {author} {\bibfnamefont {H.}~\bibnamefont
  {Araki}},\ }\bibfield  {title} {\bibinfo {title} {On a characterization of
  the state space of quantum mechanics},\ }\href
  {https://doi.org/10.1007/BF01962588} {\bibfield  {journal} {\bibinfo
  {journal} {Comm. Math. Phys.}\ }\textbf {\bibinfo {volume} {75}},\ \bibinfo
  {pages} {1} (\bibinfo {year} {1980})}\BibitemShut {NoStop}%
\bibitem [{\citenamefont {Hardy}\ and\ \citenamefont {Wootters}(2012)}]{Hardy}%
  \BibitemOpen
  \bibfield  {author} {\bibinfo {author} {\bibfnamefont {L.}~\bibnamefont
  {Hardy}}\ and\ \bibinfo {author} {\bibfnamefont {W.~K.}\ \bibnamefont
  {Wootters}},\ }\bibfield  {title} {\bibinfo {title} {Limited holism and
  real-vector-space quantum theory},\ }\href
  {https://doi.org/10.1007/s10701-011-9616-6} {\bibfield  {journal} {\bibinfo
  {journal} {Found. Phys.}\ }\textbf {\bibinfo {volume} {42}},\ \bibinfo
  {pages} {454} (\bibinfo {year} {2012})}\BibitemShut {NoStop}%
\bibitem [{\citenamefont {Baez}(2012)}]{Baez}%
  \BibitemOpen
  \bibfield  {author} {\bibinfo {author} {\bibfnamefont {J.~C.}\ \bibnamefont
  {Baez}},\ }\bibfield  {title} {\bibinfo {title} {Division algebras and
  quantum theory},\ }\href {https://doi.org/10.1007/s10701-011-9566-z}
  {\bibfield  {journal} {\bibinfo  {journal} {Found. Phys.}\ }\textbf {\bibinfo
  {volume} {42}},\ \bibinfo {pages} {819} (\bibinfo {year} {2012})}\BibitemShut
  {NoStop}%
\bibitem [{\citenamefont {Aleksandrova}\ \emph {et~al.}(2013)\citenamefont
  {Aleksandrova}, \citenamefont {Borish},\ and\ \citenamefont
  {Wootters}}]{Aleksa}%
  \BibitemOpen
  \bibfield  {author} {\bibinfo {author} {\bibfnamefont {A.}~\bibnamefont
  {Aleksandrova}}, \bibinfo {author} {\bibfnamefont {V.}~\bibnamefont
  {Borish}},\ and\ \bibinfo {author} {\bibfnamefont {W.~K.}\ \bibnamefont
  {Wootters}},\ }\bibfield  {title} {\bibinfo {title} {Real-vector-space
  quantum theory with a universal quantum bit},\ }\href
  {https://doi.org/10.1103/PhysRevA.87.052106} {\bibfield  {journal} {\bibinfo
  {journal} {Phys. Rev. A}\ }\textbf {\bibinfo {volume} {87}},\ \bibinfo
  {pages} {052106} (\bibinfo {year} {2013})}\BibitemShut {NoStop}%
\bibitem [{\citenamefont {Wootters}(2015)}]{Wootters}%
  \BibitemOpen
  \bibfield  {author} {\bibinfo {author} {\bibfnamefont {W.~K.}\ \bibnamefont
  {Wootters}},\ }\bibfield  {title} {\bibinfo {title} {Optimal information
  transfer and real-vector-space quantum theory},\ }\href
  {https://doi.org/10.1007/978-94-017-7303-4_2} {\bibfield  {journal} {\bibinfo
   {journal} {Quantum Theory: Informational Foundations and Foils, edited by G.
  Chiribella and R. W. Spekkens}\ }\textbf {\bibinfo {volume} {181}},\ \bibinfo
  {pages} {21} (\bibinfo {year} {2015})}\BibitemShut {NoStop}%
\bibitem [{\citenamefont {Liu}\ \emph {et~al.}(2019)\citenamefont {Liu},
  \citenamefont {Elliott}, \citenamefont {Binder}, \citenamefont {Di~Franco},\
  and\ \citenamefont {Gu}}]{Liu}%
  \BibitemOpen
  \bibfield  {author} {\bibinfo {author} {\bibfnamefont {Q.}~\bibnamefont
  {Liu}}, \bibinfo {author} {\bibfnamefont {T.~J.}\ \bibnamefont {Elliott}},
  \bibinfo {author} {\bibfnamefont {F.~C.}\ \bibnamefont {Binder}}, \bibinfo
  {author} {\bibfnamefont {C.}~\bibnamefont {Di~Franco}},\ and\ \bibinfo
  {author} {\bibfnamefont {M.}~\bibnamefont {Gu}},\ }\bibfield  {title}
  {\bibinfo {title} {Optimal stochastic modeling with unitary quantum
  dynamics},\ }\href {https://doi.org/10.1103/PhysRevA.99.062110} {\bibfield
  {journal} {\bibinfo  {journal} {Phys. Rev. A}\ }\textbf {\bibinfo {volume}
  {99}},\ \bibinfo {pages} {062110} (\bibinfo {year} {2019})}\BibitemShut
  {NoStop}%
\bibitem [{\citenamefont {Zhu}(2020)}]{Zhu}%
  \BibitemOpen
  \bibfield  {author} {\bibinfo {author} {\bibfnamefont {H.}~\bibnamefont
  {Zhu}},\ }\href@noop {} {\bibinfo {title} {Hiding and masking quantum
  information in complex and real quantum mechanics}} (\bibinfo {year}
  {2020}),\ \Eprint {https://arxiv.org/abs/2010.07843} {arXiv:2010.07843
  [quant-ph]} \BibitemShut {NoStop}%
\bibitem [{\citenamefont {Zhang}\ \emph {et~al.}(2021)\citenamefont {Zhang},
  \citenamefont {Hou}, \citenamefont {Li}, \citenamefont {Zhu}, \citenamefont
  {Xiang}, \citenamefont {Li},\ and\ \citenamefont {Guo}}]{Zhang}%
  \BibitemOpen
  \bibfield  {author} {\bibinfo {author} {\bibfnamefont {R.-Q.}\ \bibnamefont
  {Zhang}}, \bibinfo {author} {\bibfnamefont {Z.}~\bibnamefont {Hou}}, \bibinfo
  {author} {\bibfnamefont {Z.}~\bibnamefont {Li}}, \bibinfo {author}
  {\bibfnamefont {H.}~\bibnamefont {Zhu}}, \bibinfo {author} {\bibfnamefont
  {G.-Y.}\ \bibnamefont {Xiang}}, \bibinfo {author} {\bibfnamefont {C.-F.}\
  \bibnamefont {Li}},\ and\ \bibinfo {author} {\bibfnamefont {G.-C.}\
  \bibnamefont {Guo}},\ }\href@noop {} {\bibinfo {title} {Experimental masking
  of real quantum states}} (\bibinfo {year} {2021}),\ \Eprint
  {https://arxiv.org/abs/2107.01589} {arXiv:2107.01589 [quant-ph]} \BibitemShut
  {NoStop}%
\bibitem [{\citenamefont {Hickey}\ and\ \citenamefont {Gour}(2018)}]{Hickey}%
  \BibitemOpen
  \bibfield  {author} {\bibinfo {author} {\bibfnamefont {A.}~\bibnamefont
  {Hickey}}\ and\ \bibinfo {author} {\bibfnamefont {G.}~\bibnamefont {Gour}},\
  }\bibfield  {title} {\bibinfo {title} {Quantifying the imaginarity of quantum
  mechanics},\ }\href {https://doi.org/10.1088/1751-8121/aabe9c} {\bibfield
  {journal} {\bibinfo  {journal} {J. Phys. A: Math. Theor.}\ }\textbf {\bibinfo
  {volume} {51}},\ \bibinfo {pages} {414009} (\bibinfo {year}
  {2018})}\BibitemShut {NoStop}%
\bibitem [{\citenamefont {Wu}\ \emph {et~al.}(2021{\natexlab{a}})\citenamefont
  {Wu}, \citenamefont {Kondra}, \citenamefont {Rana}, \citenamefont {Scandolo},
  \citenamefont {Xiang}, \citenamefont {Li}, \citenamefont {Guo},\ and\
  \citenamefont {Streltsov}}]{Wu}%
  \BibitemOpen
  \bibfield  {author} {\bibinfo {author} {\bibfnamefont {K.-D.}\ \bibnamefont
  {Wu}}, \bibinfo {author} {\bibfnamefont {T.~V.}\ \bibnamefont {Kondra}},
  \bibinfo {author} {\bibfnamefont {S.}~\bibnamefont {Rana}}, \bibinfo {author}
  {\bibfnamefont {C.~M.}\ \bibnamefont {Scandolo}}, \bibinfo {author}
  {\bibfnamefont {G.-Y.}\ \bibnamefont {Xiang}}, \bibinfo {author}
  {\bibfnamefont {C.-F.}\ \bibnamefont {Li}}, \bibinfo {author} {\bibfnamefont
  {G.-C.}\ \bibnamefont {Guo}},\ and\ \bibinfo {author} {\bibfnamefont
  {A.}~\bibnamefont {Streltsov}},\ }\bibfield  {title} {\bibinfo {title}
  {Operational resource theory of imaginarity},\ }\href
  {https://doi.org/10.1103/PhysRevLett.126.090401} {\bibfield  {journal}
  {\bibinfo  {journal} {Phys. Rev. Lett.}\ }\textbf {\bibinfo {volume} {126}},\
  \bibinfo {pages} {090401} (\bibinfo {year} {2021}{\natexlab{a}})}\BibitemShut
  {NoStop}%
\bibitem [{\citenamefont {Wu}\ \emph {et~al.}(2021{\natexlab{b}})\citenamefont
  {Wu}, \citenamefont {Kondra}, \citenamefont {Rana}, \citenamefont {Scandolo},
  \citenamefont {Xiang}, \citenamefont {Li}, \citenamefont {Guo},\ and\
  \citenamefont {Streltsov}}]{Wu1}%
  \BibitemOpen
  \bibfield  {author} {\bibinfo {author} {\bibfnamefont {K.-D.}\ \bibnamefont
  {Wu}}, \bibinfo {author} {\bibfnamefont {T.~V.}\ \bibnamefont {Kondra}},
  \bibinfo {author} {\bibfnamefont {S.}~\bibnamefont {Rana}}, \bibinfo {author}
  {\bibfnamefont {C.~M.}\ \bibnamefont {Scandolo}}, \bibinfo {author}
  {\bibfnamefont {G.-Y.}\ \bibnamefont {Xiang}}, \bibinfo {author}
  {\bibfnamefont {C.-F.}\ \bibnamefont {Li}}, \bibinfo {author} {\bibfnamefont
  {G.-C.}\ \bibnamefont {Guo}},\ and\ \bibinfo {author} {\bibfnamefont
  {A.}~\bibnamefont {Streltsov}},\ }\bibfield  {title} {\bibinfo {title}
  {Resource theory of imaginarity: Quantification and state conversion},\
  }\href {https://doi.org/10.1103/PhysRevA.103.032401} {\bibfield  {journal}
  {\bibinfo  {journal} {Phys. Rev. A}\ }\textbf {\bibinfo {volume} {103}},\
  \bibinfo {pages} {032401} (\bibinfo {year} {2021}{\natexlab{b}})}\BibitemShut
  {NoStop}%
\bibitem [{\citenamefont {Wu}\ \emph {et~al.}(2024)\citenamefont {Wu},
  \citenamefont {Kondra}, \citenamefont {Scandolo}, \citenamefont {M.~Rana},
  \citenamefont {Xiang}, \citenamefont {Li}, \citenamefont {Guo},\ and\
  \citenamefont {Streltsov}}]{Wu2}%
  \BibitemOpen
  \bibfield  {author} {\bibinfo {author} {\bibfnamefont {K.-D.}\ \bibnamefont
  {Wu}}, \bibinfo {author} {\bibfnamefont {T.~V.}\ \bibnamefont {Kondra}},
  \bibinfo {author} {\bibfnamefont {C.}~\bibnamefont {Scandolo}}, \bibinfo
  {author} {\bibfnamefont {S.}~\bibnamefont {M.~Rana}}, \bibinfo {author}
  {\bibfnamefont {G.-Y.}\ \bibnamefont {Xiang}}, \bibinfo {author}
  {\bibfnamefont {C.-F.}\ \bibnamefont {Li}}, \bibinfo {author} {\bibfnamefont
  {G.-C.}\ \bibnamefont {Guo}},\ and\ \bibinfo {author} {\bibfnamefont
  {A.}~\bibnamefont {Streltsov}},\ }\bibfield  {title} {\bibinfo {title}
  {Resource theory of imaginarity in distributed scenarios},\ }\href
  {https://doi.org/https://doi.org/10.1038/s42005-024-01649-y} {\bibfield
  {journal} {\bibinfo  {journal} {Comm. Phys.}\ }\textbf {\bibinfo {volume}
  {7}},\ \bibinfo {pages} {171} (\bibinfo {year} {2024})}\BibitemShut {NoStop}%
\bibitem [{\citenamefont {Zhang}\ \emph {et~al.}(2024)\citenamefont {Zhang},
  \citenamefont {Li},\ and\ \citenamefont {Luo}}]{Zhang2024}%
  \BibitemOpen
  \bibfield  {author} {\bibinfo {author} {\bibfnamefont {Z.}~\bibnamefont
  {Zhang}}, \bibinfo {author} {\bibfnamefont {N.}~\bibnamefont {Li}},\ and\
  \bibinfo {author} {\bibfnamefont {S.}~\bibnamefont {Luo}},\ }\bibfield
  {title} {\bibinfo {title} {Broadcasting of imaginarity},\ }\href
  {https://doi.org/10.1103/PhysRevA.110.052439} {\bibfield  {journal} {\bibinfo
   {journal} {Phys. Rev. A}\ }\textbf {\bibinfo {volume} {110}},\ \bibinfo
  {pages} {052439} (\bibinfo {year} {2024})}\BibitemShut {NoStop}%
\bibitem [{\citenamefont {Zhu}(2021)}]{Zhu2021}%
  \BibitemOpen
  \bibfield  {author} {\bibinfo {author} {\bibfnamefont {H.}~\bibnamefont
  {Zhu}},\ }\bibfield  {title} {\bibinfo {title} {Hiding and masking quantum
  information in complex and real quantum mechanics},\ }\href
  {https://doi.org/10.1103/PhysRevResearch.3.033176} {\bibfield  {journal}
  {\bibinfo  {journal} {Phys. Rev. Res.}\ }\textbf {\bibinfo {volume} {3}},\
  \bibinfo {pages} {033176} (\bibinfo {year} {2021})}\BibitemShut {NoStop}%
\bibitem [{\citenamefont {Batle}\ \emph {et~al.}(2025)\citenamefont {Batle},
  \citenamefont {Białecki}, \citenamefont {Rybotycki}, \citenamefont
  {Tworzydło},\ and\ \citenamefont {Bednorz}}]{Batle2025}%
  \BibitemOpen
  \bibfield  {author} {\bibinfo {author} {\bibfnamefont {J.}~\bibnamefont
  {Batle}}, \bibinfo {author} {\bibfnamefont {T.}~\bibnamefont {Białecki}},
  \bibinfo {author} {\bibfnamefont {T.}~\bibnamefont {Rybotycki}}, \bibinfo
  {author} {\bibfnamefont {J.}~\bibnamefont {Tworzydło}},\ and\ \bibinfo
  {author} {\bibfnamefont {A.}~\bibnamefont {Bednorz}},\ }\bibfield  {title}
  {\bibinfo {title} {Efficient discrimination between real and complex quantum
  theories},\ }\href {https://doi.org/10.22331/q-2025-01-15-1595} {\bibfield
  {journal} {\bibinfo  {journal} {Quantum}\ }\textbf {\bibinfo {volume} {9}},\
  \bibinfo {pages} {1595} (\bibinfo {year} {2025})}\BibitemShut {NoStop}%
\bibitem [{\citenamefont {Einstein}\ \emph {et~al.}(1935)\citenamefont
  {Einstein}, \citenamefont {Podolsky},\ and\ \citenamefont {Rosen}}]{EPR}%
  \BibitemOpen
  \bibfield  {author} {\bibinfo {author} {\bibfnamefont {A.}~\bibnamefont
  {Einstein}}, \bibinfo {author} {\bibfnamefont {B.}~\bibnamefont {Podolsky}},\
  and\ \bibinfo {author} {\bibfnamefont {N.}~\bibnamefont {Rosen}},\ }\bibfield
   {title} {\bibinfo {title} {Can quantum-mechanical description of physical
  reality be considered complete?},\ }\href
  {https://doi.org/10.1103/PhysRev.47.777} {\bibfield  {journal} {\bibinfo
  {journal} {Phys. Rev.}\ }\textbf {\bibinfo {volume} {47}},\ \bibinfo {pages}
  {777} (\bibinfo {year} {1935})}\BibitemShut {NoStop}%
\bibitem [{\citenamefont {Bell}(1964)}]{Bell}%
  \BibitemOpen
  \bibfield  {author} {\bibinfo {author} {\bibfnamefont {J.~S.}\ \bibnamefont
  {Bell}},\ }\bibfield  {title} {\bibinfo {title} {On the einstein podolsky
  rosen paradox},\ }\href {https://doi.org/10.1103/PhysicsPhysiqueFizika.1.195}
  {\bibfield  {journal} {\bibinfo  {journal} {Physics Physique Fizika}\
  }\textbf {\bibinfo {volume} {1}},\ \bibinfo {pages} {195} (\bibinfo {year}
  {1964})}\BibitemShut {NoStop}%
\bibitem [{\citenamefont {Clauser}\ \emph {et~al.}(1969)\citenamefont
  {Clauser}, \citenamefont {Horne}, \citenamefont {Shimony},\ and\
  \citenamefont {Holt}}]{CHSH}%
  \BibitemOpen
  \bibfield  {author} {\bibinfo {author} {\bibfnamefont {J.~F.}\ \bibnamefont
  {Clauser}}, \bibinfo {author} {\bibfnamefont {M.~A.}\ \bibnamefont {Horne}},
  \bibinfo {author} {\bibfnamefont {A.}~\bibnamefont {Shimony}},\ and\ \bibinfo
  {author} {\bibfnamefont {R.~A.}\ \bibnamefont {Holt}},\ }\bibfield  {title}
  {\bibinfo {title} {Proposed experiment to test local hidden-variable
  theories},\ }\href {https://doi.org/10.1103/PhysRevLett.23.880} {\bibfield
  {journal} {\bibinfo  {journal} {Phys. Rev. Lett.}\ }\textbf {\bibinfo
  {volume} {23}},\ \bibinfo {pages} {880} (\bibinfo {year} {1969})}\BibitemShut
  {NoStop}%
\bibitem [{\citenamefont {Schrödinger}(1935)}]{Sch}%
  \BibitemOpen
  \bibfield  {author} {\bibinfo {author} {\bibfnamefont {E.}~\bibnamefont
  {Schrödinger}},\ }\bibfield  {title} {\bibinfo {title} {Discussion of
  probability relations between separated systems},\ }\href
  {https://doi.org/10.1017/S0305004100013554} {\bibfield  {journal} {\bibinfo
  {journal} {Mathematical Proceedings of the Cambridge Philosophical Society}\
  }\textbf {\bibinfo {volume} {31}},\ \bibinfo {pages} {555–563} (\bibinfo
  {year} {1935})}\BibitemShut {NoStop}%
\bibitem [{\citenamefont {Cavalcanti}\ \emph {et~al.}(2015)\citenamefont
  {Cavalcanti}, \citenamefont {Foster}, \citenamefont {Fuwa},\ and\
  \citenamefont {Wiseman}}]{CFFW}%
  \BibitemOpen
  \bibfield  {author} {\bibinfo {author} {\bibfnamefont {E.~G.}\ \bibnamefont
  {Cavalcanti}}, \bibinfo {author} {\bibfnamefont {C.~J.}\ \bibnamefont
  {Foster}}, \bibinfo {author} {\bibfnamefont {M.}~\bibnamefont {Fuwa}},\ and\
  \bibinfo {author} {\bibfnamefont {H.~M.}\ \bibnamefont {Wiseman}},\
  }\bibfield  {title} {\bibinfo {title} {Analog of the
  clauser-horne-shimony-holt inequality for steering},\ }\href
  {https://doi.org/10.1364/JOSAB.32.000A74} {\bibfield  {journal} {\bibinfo
  {journal} {J. Opt. Soc. Am. B}\ }\textbf {\bibinfo {volume} {32}},\ \bibinfo
  {pages} {A74} (\bibinfo {year} {2015})}\BibitemShut {NoStop}%
\bibitem [{\citenamefont {Wiseman}\ \emph {et~al.}(2007)\citenamefont
  {Wiseman}, \citenamefont {Jones},\ and\ \citenamefont {Doherty}}]{Wiseman}%
  \BibitemOpen
  \bibfield  {author} {\bibinfo {author} {\bibfnamefont {H.~M.}\ \bibnamefont
  {Wiseman}}, \bibinfo {author} {\bibfnamefont {S.~J.}\ \bibnamefont {Jones}},\
  and\ \bibinfo {author} {\bibfnamefont {A.~C.}\ \bibnamefont {Doherty}},\
  }\bibfield  {title} {\bibinfo {title} {Steering, entanglement, nonlocality,
  and the einstein-podolsky-rosen paradox},\ }\href
  {https://doi.org/10.1103/PhysRevLett.98.140402} {\bibfield  {journal}
  {\bibinfo  {journal} {Phys. Rev. Lett.}\ }\textbf {\bibinfo {volume} {98}},\
  \bibinfo {pages} {140402} (\bibinfo {year} {2007})}\BibitemShut {NoStop}%
\bibitem [{\citenamefont {Jones}\ \emph {et~al.}(2007)\citenamefont {Jones},
  \citenamefont {Wiseman},\ and\ \citenamefont {Doherty}}]{Jones}%
  \BibitemOpen
  \bibfield  {author} {\bibinfo {author} {\bibfnamefont {S.~J.}\ \bibnamefont
  {Jones}}, \bibinfo {author} {\bibfnamefont {H.~M.}\ \bibnamefont {Wiseman}},\
  and\ \bibinfo {author} {\bibfnamefont {A.~C.}\ \bibnamefont {Doherty}},\
  }\bibfield  {title} {\bibinfo {title} {Entanglement, einstein-podolsky-rosen
  correlations, bell nonlocality, and steering},\ }\href
  {https://doi.org/10.1103/PhysRevA.76.052116} {\bibfield  {journal} {\bibinfo
  {journal} {Phys. Rev. A}\ }\textbf {\bibinfo {volume} {76}},\ \bibinfo
  {pages} {052116} (\bibinfo {year} {2007})}\BibitemShut {NoStop}%
\bibitem [{\citenamefont {{\v{S}}upi{\'{c}}}\ and\ \citenamefont
  {Bowles}(2020)}]{Supic2020}%
  \BibitemOpen
  \bibfield  {author} {\bibinfo {author} {\bibfnamefont {I.}~\bibnamefont
  {{\v{S}}upi{\'{c}}}}\ and\ \bibinfo {author} {\bibfnamefont {J.}~\bibnamefont
  {Bowles}},\ }\bibfield  {title} {\bibinfo {title} {Self-testing of quantum
  systems: a review},\ }\href {https://doi.org/10.22331/q-2020-09-30-337}
  {\bibfield  {journal} {\bibinfo  {journal} {{Quantum}}\ }\textbf {\bibinfo
  {volume} {4}},\ \bibinfo {pages} {337} (\bibinfo {year} {2020})}\BibitemShut
  {NoStop}%
\bibitem [{\citenamefont {Bian}\ \emph {et~al.}(2020)\citenamefont {Bian},
  \citenamefont {Majumdar}, \citenamefont {Jebarathinam}, \citenamefont {Wang},
  \citenamefont {Xiao}, \citenamefont {Zhan}, \citenamefont {Zhang},\ and\
  \citenamefont {Xue}}]{Bian2020}%
  \BibitemOpen
  \bibfield  {author} {\bibinfo {author} {\bibfnamefont {Z.}~\bibnamefont
  {Bian}}, \bibinfo {author} {\bibfnamefont {A.~S.}\ \bibnamefont {Majumdar}},
  \bibinfo {author} {\bibfnamefont {C.}~\bibnamefont {Jebarathinam}}, \bibinfo
  {author} {\bibfnamefont {K.}~\bibnamefont {Wang}}, \bibinfo {author}
  {\bibfnamefont {L.}~\bibnamefont {Xiao}}, \bibinfo {author} {\bibfnamefont
  {X.}~\bibnamefont {Zhan}}, \bibinfo {author} {\bibfnamefont {Y.}~\bibnamefont
  {Zhang}},\ and\ \bibinfo {author} {\bibfnamefont {P.}~\bibnamefont {Xue}},\
  }\bibfield  {title} {\bibinfo {title} {Experimental demonstration of
  one-sided device-independent self-testing of any pure two-qubit entangled
  state},\ }\href {https://doi.org/10.1103/PhysRevA.101.020301} {\bibfield
  {journal} {\bibinfo  {journal} {Phys. Rev. A}\ }\textbf {\bibinfo {volume}
  {101}},\ \bibinfo {pages} {020301(R)} (\bibinfo {year} {2020})}\BibitemShut
  {NoStop}%
\bibitem [{\citenamefont {Uola}\ \emph {et~al.}(2020)\citenamefont {Uola},
  \citenamefont {Costa}, \citenamefont {Nguyen},\ and\ \citenamefont
  {G\"{u}hne}}]{Uola}%
  \BibitemOpen
  \bibfield  {author} {\bibinfo {author} {\bibfnamefont {R.}~\bibnamefont
  {Uola}}, \bibinfo {author} {\bibfnamefont {A.~C.~S.}\ \bibnamefont {Costa}},
  \bibinfo {author} {\bibfnamefont {H.~C.}\ \bibnamefont {Nguyen}},\ and\
  \bibinfo {author} {\bibfnamefont {O.}~\bibnamefont {G\"{u}hne}},\ }\bibfield
  {title} {\bibinfo {title} {Quantum steering},\ }\href
  {https://doi.org/10.1103/RevModPhys.92.015001} {\bibfield  {journal}
  {\bibinfo  {journal} {Rev. Mod. Phys.}\ }\textbf {\bibinfo {volume} {92}},\
  \bibinfo {pages} {015001} (\bibinfo {year} {2020})}\BibitemShut {NoStop}%
\bibitem [{\citenamefont {Reid}(1989)}]{Reid}%
  \BibitemOpen
  \bibfield  {author} {\bibinfo {author} {\bibfnamefont {M.~D.}\ \bibnamefont
  {Reid}},\ }\bibfield  {title} {\bibinfo {title} {Demonstration of the
  einstein-podolsky-rosen paradox using nondegenerate parametric
  amplification},\ }\href {https://doi.org/10.1103/PhysRevA.40.913} {\bibfield
  {journal} {\bibinfo  {journal} {Phys. Rev. A}\ }\textbf {\bibinfo {volume}
  {40}},\ \bibinfo {pages} {913} (\bibinfo {year} {1989})}\BibitemShut
  {NoStop}%
\bibitem [{\citenamefont {Schneeloch}\ \emph {et~al.}(2013)\citenamefont
  {Schneeloch}, \citenamefont {Broadbent}, \citenamefont {Walborn},
  \citenamefont {Cavalcanti},\ and\ \citenamefont {Howell}}]{Schneeloch}%
  \BibitemOpen
  \bibfield  {author} {\bibinfo {author} {\bibfnamefont {J.}~\bibnamefont
  {Schneeloch}}, \bibinfo {author} {\bibfnamefont {C.~J.}\ \bibnamefont
  {Broadbent}}, \bibinfo {author} {\bibfnamefont {S.~P.}\ \bibnamefont
  {Walborn}}, \bibinfo {author} {\bibfnamefont {E.~G.}\ \bibnamefont
  {Cavalcanti}},\ and\ \bibinfo {author} {\bibfnamefont {J.~C.}\ \bibnamefont
  {Howell}},\ }\bibfield  {title} {\bibinfo {title} {Einstein-podolsky-rosen
  steering inequalities from entropic uncertainty relations},\ }\href
  {https://doi.org/10.1103/PhysRevA.87.062103} {\bibfield  {journal} {\bibinfo
  {journal} {Phys. Rev. A}\ }\textbf {\bibinfo {volume} {87}},\ \bibinfo
  {pages} {062103} (\bibinfo {year} {2013})}\BibitemShut {NoStop}%
\bibitem [{\citenamefont {Pramanik}\ \emph
  {et~al.}(2014{\natexlab{a}})\citenamefont {Pramanik}, \citenamefont
  {Kaplan},\ and\ \citenamefont {Majumdar}}]{Pramanik}%
  \BibitemOpen
  \bibfield  {author} {\bibinfo {author} {\bibfnamefont {T.}~\bibnamefont
  {Pramanik}}, \bibinfo {author} {\bibfnamefont {M.}~\bibnamefont {Kaplan}},\
  and\ \bibinfo {author} {\bibfnamefont {A.~S.}\ \bibnamefont {Majumdar}},\
  }\bibfield  {title} {\bibinfo {title} {Fine-grained
  einstein-podolsky-rosen--steering inequalities},\ }\href
  {https://doi.org/10.1103/PhysRevA.90.050305} {\bibfield  {journal} {\bibinfo
  {journal} {Phys. Rev. A}\ }\textbf {\bibinfo {volume} {90}},\ \bibinfo
  {pages} {050305} (\bibinfo {year} {2014}{\natexlab{a}})}\BibitemShut
  {NoStop}%
\bibitem [{\citenamefont {Maity}\ \emph {et~al.}(2017)\citenamefont {Maity},
  \citenamefont {Datta},\ and\ \citenamefont {Majumdar}}]{Maity}%
  \BibitemOpen
  \bibfield  {author} {\bibinfo {author} {\bibfnamefont {A.~G.}\ \bibnamefont
  {Maity}}, \bibinfo {author} {\bibfnamefont {S.}~\bibnamefont {Datta}},\ and\
  \bibinfo {author} {\bibfnamefont {A.~S.}\ \bibnamefont {Majumdar}},\
  }\bibfield  {title} {\bibinfo {title} {Tighter einstein-podolsky-rosen
  steering inequality based on the sum-uncertainty relation},\ }\href
  {https://doi.org/10.1103/PhysRevA.96.052326} {\bibfield  {journal} {\bibinfo
  {journal} {Phys. Rev. A}\ }\textbf {\bibinfo {volume} {96}},\ \bibinfo
  {pages} {052326} (\bibinfo {year} {2017})}\BibitemShut {NoStop}%
\bibitem [{\citenamefont {Heisenberg}(1927)}]{Heisenberg}%
  \BibitemOpen
  \bibfield  {author} {\bibinfo {author} {\bibfnamefont {W.}~\bibnamefont
  {Heisenberg}},\ }\bibfield  {title} {\bibinfo {title} {\"{U}ber den
  anschaulichen inhalt der quantentheoretischen kinematik und mechanik},\
  }\href {https://doi.org/10.1007/BF01397280} {\bibfield  {journal} {\bibinfo
  {journal} {Zeitschrift für Physik}\ }\textbf {\bibinfo {volume} {43}},\
  \bibinfo {pages} {172} (\bibinfo {year} {1927})}\BibitemShut {NoStop}%
\bibitem [{\citenamefont {Maassen}\ and\ \citenamefont
  {Uffink}(1988)}]{Massen}%
  \BibitemOpen
  \bibfield  {author} {\bibinfo {author} {\bibfnamefont {H.}~\bibnamefont
  {Maassen}}\ and\ \bibinfo {author} {\bibfnamefont {J.~B.~M.}\ \bibnamefont
  {Uffink}},\ }\bibfield  {title} {\bibinfo {title} {Generalized entropic
  uncertainty relations},\ }\href {https://doi.org/10.1103/PhysRevLett.60.1103}
  {\bibfield  {journal} {\bibinfo  {journal} {Phys. Rev. Lett.}\ }\textbf
  {\bibinfo {volume} {60}},\ \bibinfo {pages} {1103} (\bibinfo {year}
  {1988})}\BibitemShut {NoStop}%
\bibitem [{\citenamefont {Oppenheim}\ and\ \citenamefont
  {Wehner}(2010)}]{Oppenheim}%
  \BibitemOpen
  \bibfield  {author} {\bibinfo {author} {\bibfnamefont {J.}~\bibnamefont
  {Oppenheim}}\ and\ \bibinfo {author} {\bibfnamefont {S.}~\bibnamefont
  {Wehner}},\ }\bibfield  {title} {\bibinfo {title} {The uncertainty principle
  determines the nonlocality of quantum mechanics},\ }\href
  {https://doi.org/10.1126/science.1192065} {\bibfield  {journal} {\bibinfo
  {journal} {Science}\ }\textbf {\bibinfo {volume} {330}},\ \bibinfo {pages}
  {1072} (\bibinfo {year} {2010})}\BibitemShut {NoStop}%
\bibitem [{\citenamefont {Pati}\ and\ \citenamefont {Sahu}(2007)}]{Pati}%
  \BibitemOpen
  \bibfield  {author} {\bibinfo {author} {\bibfnamefont {A.}~\bibnamefont
  {Pati}}\ and\ \bibinfo {author} {\bibfnamefont {P.}~\bibnamefont {Sahu}},\
  }\bibfield  {title} {\bibinfo {title} {Sum uncertainty relation in quantum
  theory},\ }\href
  {https://doi.org/https://doi.org/10.1016/j.physleta.2007.03.005} {\bibfield
  {journal} {\bibinfo  {journal} {Phys. Lett. A}\ }\textbf {\bibinfo {volume}
  {367}},\ \bibinfo {pages} {177} (\bibinfo {year} {2007})}\BibitemShut
  {NoStop}%
\bibitem [{\citenamefont {Mondal}\ \emph {et~al.}(2017)\citenamefont {Mondal},
  \citenamefont {Pramanik},\ and\ \citenamefont {Pati}}]{Mondal}%
  \BibitemOpen
  \bibfield  {author} {\bibinfo {author} {\bibfnamefont {D.}~\bibnamefont
  {Mondal}}, \bibinfo {author} {\bibfnamefont {T.}~\bibnamefont {Pramanik}},\
  and\ \bibinfo {author} {\bibfnamefont {A.~K.}\ \bibnamefont {Pati}},\
  }\bibfield  {title} {\bibinfo {title} {Nonlocal advantage of quantum
  coherence},\ }\href {https://doi.org/10.1103/PhysRevA.95.010301} {\bibfield
  {journal} {\bibinfo  {journal} {Phys. Rev. A}\ }\textbf {\bibinfo {volume}
  {95}},\ \bibinfo {pages} {010301} (\bibinfo {year} {2017})}\BibitemShut
  {NoStop}%
\bibitem [{\citenamefont {Wei}\ and\ \citenamefont {Fei}(2024)}]{Wei2024}%
  \BibitemOpen
  \bibfield  {author} {\bibinfo {author} {\bibfnamefont {Z.-W.}\ \bibnamefont
  {Wei}}\ and\ \bibinfo {author} {\bibfnamefont {S.-M.}\ \bibnamefont {Fei}},\
  }\bibfield  {title} {\bibinfo {title} {Nonlocal advantages of quantum
  imaginarity},\ }\href {https://doi.org/10.1103/PhysRevA.110.052202}
  {\bibfield  {journal} {\bibinfo  {journal} {Phys. Rev. A}\ }\textbf {\bibinfo
  {volume} {110}},\ \bibinfo {pages} {052202} (\bibinfo {year}
  {2024})}\BibitemShut {NoStop}%
\bibitem [{\citenamefont {Bengtsson}(2007)}]{MUB}%
  \BibitemOpen
  \bibfield  {author} {\bibinfo {author} {\bibfnamefont {I.}~\bibnamefont
  {Bengtsson}},\ }\bibfield  {title} {\bibinfo {title} {Three ways to look at
  mutually unbiased bases},\ }\href {https://doi.org/10.1063/1.2713445}
  {\bibfield  {journal} {\bibinfo  {journal} {AIP Conference Proceedings}\
  }\textbf {\bibinfo {volume} {889}},\ \bibinfo {pages} {40} (\bibinfo {year}
  {2007})},\ \Eprint
  {https://arxiv.org/abs/https://aip.scitation.org/doi/pdf/10.1063/1.2713445}
  {https://aip.scitation.org/doi/pdf/10.1063/1.2713445} \BibitemShut {NoStop}%
\bibitem [{\citenamefont {Werner}(1989)}]{Werner}%
  \BibitemOpen
  \bibfield  {author} {\bibinfo {author} {\bibfnamefont {R.~F.}\ \bibnamefont
  {Werner}},\ }\bibfield  {title} {\bibinfo {title} {Quantum states with
  einstein-podolsky-rosen correlations admitting a hidden-variable model},\
  }\href {https://doi.org/10.1103/PhysRevA.40.4277} {\bibfield  {journal}
  {\bibinfo  {journal} {Phys. Rev. A}\ }\textbf {\bibinfo {volume} {40}},\
  \bibinfo {pages} {4277} (\bibinfo {year} {1989})}\BibitemShut {NoStop}%
\bibitem [{\citenamefont {Rau}(2009)}]{Rau2009}%
  \BibitemOpen
  \bibfield  {author} {\bibinfo {author} {\bibfnamefont {A.~R.~P.}\
  \bibnamefont {Rau}},\ }\bibfield  {title} {\bibinfo {title} {Algebraic
  characterization of x-states in quantum information},\ }\href
  {https://doi.org/10.1088/1751-8113/42/41/412002} {\bibfield  {journal}
  {\bibinfo  {journal} {J. Phys. A: Math. Theor.}\ }\textbf {\bibinfo {volume}
  {42}},\ \bibinfo {pages} {412002} (\bibinfo {year} {2009})}\BibitemShut
  {NoStop}%
\bibitem [{\citenamefont {Ac\'{\i}n}\ \emph {et~al.}(2000)\citenamefont
  {Ac\'{\i}n}, \citenamefont {Andrianov}, \citenamefont {Costa}, \citenamefont
  {Jan\'e}, \citenamefont {Latorre},\ and\ \citenamefont {Tarrach}}]{Acin2000}%
  \BibitemOpen
  \bibfield  {author} {\bibinfo {author} {\bibfnamefont {A.}~\bibnamefont
  {Ac\'{\i}n}}, \bibinfo {author} {\bibfnamefont {A.}~\bibnamefont
  {Andrianov}}, \bibinfo {author} {\bibfnamefont {L.}~\bibnamefont {Costa}},
  \bibinfo {author} {\bibfnamefont {E.}~\bibnamefont {Jan\'e}}, \bibinfo
  {author} {\bibfnamefont {J.~I.}\ \bibnamefont {Latorre}},\ and\ \bibinfo
  {author} {\bibfnamefont {R.}~\bibnamefont {Tarrach}},\ }\bibfield  {title}
  {\bibinfo {title} {Generalized schmidt decomposition and classification of
  three-quantum-bit states},\ }\href
  {https://doi.org/10.1103/PhysRevLett.85.1560} {\bibfield  {journal} {\bibinfo
   {journal} {Phys. Rev. Lett.}\ }\textbf {\bibinfo {volume} {85}},\ \bibinfo
  {pages} {1560} (\bibinfo {year} {2000})}\BibitemShut {NoStop}%
\bibitem [{\citenamefont {Dhar}\ \emph {et~al.}(2017)\citenamefont {Dhar},
  \citenamefont {Pal}, \citenamefont {Rakshit}, \citenamefont {Sen(De)},\ and\
  \citenamefont {Sen}}]{Dhar2017}%
  \BibitemOpen
  \bibfield  {author} {\bibinfo {author} {\bibfnamefont {H.~S.}\ \bibnamefont
  {Dhar}}, \bibinfo {author} {\bibfnamefont {A.~K.}\ \bibnamefont {Pal}},
  \bibinfo {author} {\bibfnamefont {D.}~\bibnamefont {Rakshit}}, \bibinfo
  {author} {\bibfnamefont {A.}~\bibnamefont {Sen(De)}},\ and\ \bibinfo {author}
  {\bibfnamefont {U.}~\bibnamefont {Sen}},\ }\bibinfo {title} {Monogamy of
  quantum correlations - a review},\ in\ \href
  {https://doi.org/10.1007/978-3-319-53412-1_3} {\emph {\bibinfo {booktitle}
  {Lectures on General Quantum Correlations and their Applications}}},\
  \bibinfo {editor} {edited by\ \bibinfo {editor} {\bibfnamefont {F.~F.}\
  \bibnamefont {Fanchini}}, \bibinfo {editor} {\bibfnamefont {D.~d.~O.}\
  \bibnamefont {Soares~Pinto}},\ and\ \bibinfo {editor} {\bibfnamefont
  {G.}~\bibnamefont {Adesso}}}\ (\bibinfo  {publisher} {Springer International
  Publishing},\ \bibinfo {address} {Cham},\ \bibinfo {year} {2017})\ pp.\
  \bibinfo {pages} {23--64}\BibitemShut {NoStop}%
\bibitem [{\citenamefont {Baumgratz}\ \emph {et~al.}(2014)\citenamefont
  {Baumgratz}, \citenamefont {Cramer},\ and\ \citenamefont
  {Plenio}}]{Baumgratz}%
  \BibitemOpen
  \bibfield  {author} {\bibinfo {author} {\bibfnamefont {T.}~\bibnamefont
  {Baumgratz}}, \bibinfo {author} {\bibfnamefont {M.}~\bibnamefont {Cramer}},\
  and\ \bibinfo {author} {\bibfnamefont {M.~B.}\ \bibnamefont {Plenio}},\
  }\bibfield  {title} {\bibinfo {title} {Quantifying coherence},\ }\href
  {https://doi.org/10.1103/PhysRevLett.113.140401} {\bibfield  {journal}
  {\bibinfo  {journal} {Phys. Rev. Lett.}\ }\textbf {\bibinfo {volume} {113}},\
  \bibinfo {pages} {140401} (\bibinfo {year} {2014})}\BibitemShut {NoStop}%
\bibitem [{\citenamefont {Girolami}(2014)}]{Girolami}%
  \BibitemOpen
  \bibfield  {author} {\bibinfo {author} {\bibfnamefont {D.}~\bibnamefont
  {Girolami}},\ }\bibfield  {title} {\bibinfo {title} {Observable measure of
  quantum coherence in finite dimensional systems},\ }\href
  {https://doi.org/10.1103/PhysRevLett.113.170401} {\bibfield  {journal}
  {\bibinfo  {journal} {Phys. Rev. Lett.}\ }\textbf {\bibinfo {volume} {113}},\
  \bibinfo {pages} {170401} (\bibinfo {year} {2014})}\BibitemShut {NoStop}%
\bibitem [{\citenamefont {Winter}\ and\ \citenamefont {Yang}(2016)}]{Winter}%
  \BibitemOpen
  \bibfield  {author} {\bibinfo {author} {\bibfnamefont {A.}~\bibnamefont
  {Winter}}\ and\ \bibinfo {author} {\bibfnamefont {D.}~\bibnamefont {Yang}},\
  }\bibfield  {title} {\bibinfo {title} {Operational resource theory of
  coherence},\ }\href {https://doi.org/10.1103/PhysRevLett.116.120404}
  {\bibfield  {journal} {\bibinfo  {journal} {Phys. Rev. Lett.}\ }\textbf
  {\bibinfo {volume} {116}},\ \bibinfo {pages} {120404} (\bibinfo {year}
  {2016})}\BibitemShut {NoStop}%
\bibitem [{\citenamefont {Brand\~ao}\ and\ \citenamefont
  {Gour}(2015)}]{Brandao}%
  \BibitemOpen
  \bibfield  {author} {\bibinfo {author} {\bibfnamefont {F.~G. S.~L.}\
  \bibnamefont {Brand\~ao}}\ and\ \bibinfo {author} {\bibfnamefont
  {G.}~\bibnamefont {Gour}},\ }\bibfield  {title} {\bibinfo {title} {Reversible
  framework for quantum resource theories},\ }\href
  {https://doi.org/10.1103/PhysRevLett.115.070503} {\bibfield  {journal}
  {\bibinfo  {journal} {Phys. Rev. Lett.}\ }\textbf {\bibinfo {volume} {115}},\
  \bibinfo {pages} {070503} (\bibinfo {year} {2015})}\BibitemShut {NoStop}%
\bibitem [{\citenamefont {Fano}(1983)}]{Fano}%
  \BibitemOpen
  \bibfield  {author} {\bibinfo {author} {\bibfnamefont {U.}~\bibnamefont
  {Fano}},\ }\bibfield  {title} {\bibinfo {title} {Pairs of two-level
  systems},\ }\href {https://doi.org/10.1103/RevModPhys.55.855} {\bibfield
  {journal} {\bibinfo  {journal} {Rev. Mod. Phys.}\ }\textbf {\bibinfo {volume}
  {55}},\ \bibinfo {pages} {855} (\bibinfo {year} {1983})}\BibitemShut
  {NoStop}%
\bibitem [{\citenamefont {Horodecki}\ \emph {et~al.}(1995)\citenamefont
  {Horodecki}, \citenamefont {Horodecki},\ and\ \citenamefont
  {Horodecki}}]{Horodecki1995}%
  \BibitemOpen
  \bibfield  {author} {\bibinfo {author} {\bibfnamefont {R.}~\bibnamefont
  {Horodecki}}, \bibinfo {author} {\bibfnamefont {P.}~\bibnamefont
  {Horodecki}},\ and\ \bibinfo {author} {\bibfnamefont {M.}~\bibnamefont
  {Horodecki}},\ }\bibfield  {title} {\bibinfo {title} {Violating bell
  inequality by mixed spin-12 states: necessary and sufficient condition},\
  }\href {https://doi.org/https://doi.org/10.1016/0375-9601(95)00214-N}
  {\bibfield  {journal} {\bibinfo  {journal} {Phys. Lett. A}\ }\textbf
  {\bibinfo {volume} {200}},\ \bibinfo {pages} {340} (\bibinfo {year}
  {1995})}\BibitemShut {NoStop}%
\bibitem [{\citenamefont {Horodecki}\ \emph {et~al.}(1996)\citenamefont
  {Horodecki}, \citenamefont {Horodecki},\ and\ \citenamefont
  {Horodecki}}]{Horodecki1996}%
  \BibitemOpen
  \bibfield  {author} {\bibinfo {author} {\bibfnamefont {R.}~\bibnamefont
  {Horodecki}}, \bibinfo {author} {\bibfnamefont {M.}~\bibnamefont
  {Horodecki}},\ and\ \bibinfo {author} {\bibfnamefont {P.}~\bibnamefont
  {Horodecki}},\ }\bibfield  {title} {\bibinfo {title} {Teleportation, bell's
  inequalities and inseparability},\ }\href
  {https://doi.org/https://doi.org/10.1016/0375-9601(96)00639-1} {\bibfield
  {journal} {\bibinfo  {journal} {Phys. Lett. A}\ }\textbf {\bibinfo {volume}
  {222}},\ \bibinfo {pages} {21} (\bibinfo {year} {1996})}\BibitemShut
  {NoStop}%
\bibitem [{\citenamefont {Holmes}(1975)}]{Holmes}%
  \BibitemOpen
  \bibfield  {author} {\bibinfo {author} {\bibfnamefont {R.~B.}\ \bibnamefont
  {Holmes}},\ }\href
  {https://doi.org/https://doi.org/10.1007/978-1-4684-9369-6} {\emph {\bibinfo
  {title} {Geometric Functional Analysis and its Applications}}},\ \bibinfo
  {edition} {1st}\ ed.,\ Graduate Texts in Mathematics\ (\bibinfo  {publisher}
  {Springer Verlag},\ \bibinfo {year} {1975})\BibitemShut {NoStop}%
\bibitem [{\citenamefont {Rudin}(1976)}]{Rudin}%
  \BibitemOpen
  \bibfield  {author} {\bibinfo {author} {\bibfnamefont {W.}~\bibnamefont
  {Rudin}},\ }\href@noop {} {\emph {\bibinfo {title} {Principles of
  mathematical analysis}}},\ \bibinfo {edition} {3rd}\ ed.\ (\bibinfo
  {publisher} {New York: McGraw-Hill},\ \bibinfo {year} {1976})\BibitemShut
  {NoStop}%
\bibitem [{\citenamefont {Yu}\ and\ \citenamefont {Eberly}(2007)}]{Yu2007}%
  \BibitemOpen
  \bibfield  {author} {\bibinfo {author} {\bibfnamefont {T.}~\bibnamefont
  {Yu}}\ and\ \bibinfo {author} {\bibfnamefont {J.~H.}\ \bibnamefont
  {Eberly}},\ }\bibfield  {title} {\bibinfo {title} {Evolution from
  entanglement to decoherence of bipartite mixed ''x'' states},\ }\href
  {https://doi.org/10.48550/arXiv.quant-ph/0503089} {\bibfield  {journal}
  {\bibinfo  {journal} {Quantum Information and Computation}\ }\textbf
  {\bibinfo {volume} {7}},\ \bibinfo {pages} {459} (\bibinfo {year}
  {2007})}\BibitemShut {NoStop}%
\bibitem [{\citenamefont {Ali}\ \emph {et~al.}(2010)\citenamefont {Ali},
  \citenamefont {Rau},\ and\ \citenamefont {Alber}}]{Ali2010}%
  \BibitemOpen
  \bibfield  {author} {\bibinfo {author} {\bibfnamefont {M.}~\bibnamefont
  {Ali}}, \bibinfo {author} {\bibfnamefont {A.~R.~P.}\ \bibnamefont {Rau}},\
  and\ \bibinfo {author} {\bibfnamefont {G.}~\bibnamefont {Alber}},\ }\bibfield
   {title} {\bibinfo {title} {Quantum discord for two-qubit $x$ states},\
  }\href {https://doi.org/10.1103/PhysRevA.81.042105} {\bibfield  {journal}
  {\bibinfo  {journal} {Phys. Rev. A}\ }\textbf {\bibinfo {volume} {81}},\
  \bibinfo {pages} {042105} (\bibinfo {year} {2010})}\BibitemShut {NoStop}%
\bibitem [{\citenamefont {Kelleher}\ \emph {et~al.}(2021)\citenamefont
  {Kelleher}, \citenamefont {Holweck}, \citenamefont {Lévay},\ and\
  \citenamefont {Saniga}}]{Kelleher2021}%
  \BibitemOpen
  \bibfield  {author} {\bibinfo {author} {\bibfnamefont {C.}~\bibnamefont
  {Kelleher}}, \bibinfo {author} {\bibfnamefont {F.}~\bibnamefont {Holweck}},
  \bibinfo {author} {\bibfnamefont {P.}~\bibnamefont {Lévay}},\ and\ \bibinfo
  {author} {\bibfnamefont {M.}~\bibnamefont {Saniga}},\ }\bibfield  {title}
  {\bibinfo {title} {X-states from a finite geometric perspective},\ }\href
  {https://doi.org/https://doi.org/10.1016/j.rinp.2021.103859} {\bibfield
  {journal} {\bibinfo  {journal} {Results in Physics}\ }\textbf {\bibinfo
  {volume} {22}},\ \bibinfo {pages} {103859} (\bibinfo {year}
  {2021})}\BibitemShut {NoStop}%
\bibitem [{\citenamefont {Munro}\ \emph {et~al.}(2001)\citenamefont {Munro},
  \citenamefont {James}, \citenamefont {White},\ and\ \citenamefont
  {Kwiat}}]{Munro}%
  \BibitemOpen
  \bibfield  {author} {\bibinfo {author} {\bibfnamefont {W.~J.}\ \bibnamefont
  {Munro}}, \bibinfo {author} {\bibfnamefont {D.~F.~V.}\ \bibnamefont {James}},
  \bibinfo {author} {\bibfnamefont {A.~G.}\ \bibnamefont {White}},\ and\
  \bibinfo {author} {\bibfnamefont {P.~G.}\ \bibnamefont {Kwiat}},\ }\bibfield
  {title} {\bibinfo {title} {Maximizing the entanglement of two mixed qubits},\
  }\href {https://doi.org/10.1103/PhysRevA.64.030302} {\bibfield  {journal}
  {\bibinfo  {journal} {Phys. Rev. A}\ }\textbf {\bibinfo {volume} {64}},\
  \bibinfo {pages} {030302} (\bibinfo {year} {2001})}\BibitemShut {NoStop}%
\bibitem [{\citenamefont {Bose}\ and\ \citenamefont {Vedral}(2000)}]{Bose}%
  \BibitemOpen
  \bibfield  {author} {\bibinfo {author} {\bibfnamefont {S.}~\bibnamefont
  {Bose}}\ and\ \bibinfo {author} {\bibfnamefont {V.}~\bibnamefont {Vedral}},\
  }\bibfield  {title} {\bibinfo {title} {Mixedness and teleportation},\ }\href
  {https://doi.org/10.1103/PhysRevA.61.040101} {\bibfield  {journal} {\bibinfo
  {journal} {Phys. Rev. A}\ }\textbf {\bibinfo {volume} {61}},\ \bibinfo
  {pages} {040101} (\bibinfo {year} {2000})}\BibitemShut {NoStop}%
\bibitem [{\citenamefont {Ishizaka}\ and\ \citenamefont
  {Hiroshima}(2000)}]{Ishizaka}%
  \BibitemOpen
  \bibfield  {author} {\bibinfo {author} {\bibfnamefont {S.}~\bibnamefont
  {Ishizaka}}\ and\ \bibinfo {author} {\bibfnamefont {T.}~\bibnamefont
  {Hiroshima}},\ }\bibfield  {title} {\bibinfo {title} {Maximally entangled
  mixed states under nonlocal unitary operations in two qubits},\ }\href
  {https://doi.org/10.1103/PhysRevA.62.022310} {\bibfield  {journal} {\bibinfo
  {journal} {Phys. Rev. A}\ }\textbf {\bibinfo {volume} {62}},\ \bibinfo
  {pages} {022310} (\bibinfo {year} {2000})}\BibitemShut {NoStop}%
\bibitem [{\citenamefont {Verstraete}\ \emph {et~al.}(2001)\citenamefont
  {Verstraete}, \citenamefont {Audenaert},\ and\ \citenamefont
  {De~Moor}}]{Vers}%
  \BibitemOpen
  \bibfield  {author} {\bibinfo {author} {\bibfnamefont {F.}~\bibnamefont
  {Verstraete}}, \bibinfo {author} {\bibfnamefont {K.}~\bibnamefont
  {Audenaert}},\ and\ \bibinfo {author} {\bibfnamefont {B.}~\bibnamefont
  {De~Moor}},\ }\bibfield  {title} {\bibinfo {title} {Maximally entangled mixed
  states of two qubits},\ }\href {https://doi.org/10.1103/PhysRevA.64.012316}
  {\bibfield  {journal} {\bibinfo  {journal} {Phys. Rev. A}\ }\textbf {\bibinfo
  {volume} {64}},\ \bibinfo {pages} {012316} (\bibinfo {year}
  {2001})}\BibitemShut {NoStop}%
\bibitem [{\citenamefont {Coffman}\ \emph
  {et~al.}(2000{\natexlab{a}})\citenamefont {Coffman}, \citenamefont {Kundu},\
  and\ \citenamefont {Wootters}}]{Coffman}%
  \BibitemOpen
  \bibfield  {author} {\bibinfo {author} {\bibfnamefont {V.}~\bibnamefont
  {Coffman}}, \bibinfo {author} {\bibfnamefont {J.}~\bibnamefont {Kundu}},\
  and\ \bibinfo {author} {\bibfnamefont {W.~K.}\ \bibnamefont {Wootters}},\
  }\bibfield  {title} {\bibinfo {title} {Distributed entanglement},\ }\href
  {https://doi.org/10.1103/PhysRevA.61.052306} {\bibfield  {journal} {\bibinfo
  {journal} {Phys. Rev. A}\ }\textbf {\bibinfo {volume} {61}},\ \bibinfo
  {pages} {052306} (\bibinfo {year} {2000}{\natexlab{a}})}\BibitemShut
  {NoStop}%
\bibitem [{\citenamefont {Fernandes}\ \emph {et~al.}(2024)\citenamefont
  {Fernandes}, \citenamefont {Wagner}, \citenamefont {Novo},\ and\
  \citenamefont {Galv\~ao}}]{Fernandes2024}%
  \BibitemOpen
  \bibfield  {author} {\bibinfo {author} {\bibfnamefont {C.}~\bibnamefont
  {Fernandes}}, \bibinfo {author} {\bibfnamefont {R.}~\bibnamefont {Wagner}},
  \bibinfo {author} {\bibfnamefont {L.}~\bibnamefont {Novo}},\ and\ \bibinfo
  {author} {\bibfnamefont {E.~F.}\ \bibnamefont {Galv\~ao}},\ }\bibfield
  {title} {\bibinfo {title} {Unitary-invariant witnesses of quantum
  imaginarity},\ }\href {https://doi.org/10.1103/PhysRevLett.133.190201}
  {\bibfield  {journal} {\bibinfo  {journal} {Phys. Rev. Lett.}\ }\textbf
  {\bibinfo {volume} {133}},\ \bibinfo {pages} {190201} (\bibinfo {year}
  {2024})}\BibitemShut {NoStop}%
\bibitem [{\citenamefont {Zhang}\ and\ \citenamefont {Li}(2025)}]{Zhang2025}%
  \BibitemOpen
  \bibfield  {author} {\bibinfo {author} {\bibfnamefont {L.}~\bibnamefont
  {Zhang}}\ and\ \bibinfo {author} {\bibfnamefont {N.}~\bibnamefont {Li}},\
  }\bibfield  {title} {\bibinfo {title} {On imaginarity witnesses},\ }\href
  {https://doi.org/https://doi.org/10.1016/j.physleta.2024.130135} {\bibfield
  {journal} {\bibinfo  {journal} {Physics Letters A}\ }\textbf {\bibinfo
  {volume} {530}},\ \bibinfo {pages} {130135} (\bibinfo {year}
  {2025})}\BibitemShut {NoStop}%
\bibitem [{\citenamefont {G\"uhne}\ \emph {et~al.}(2002)\citenamefont
  {G\"uhne}, \citenamefont {Hyllus}, \citenamefont {Bru\ss{}}, \citenamefont
  {Ekert}, \citenamefont {Lewenstein}, \citenamefont {Macchiavello},\ and\
  \citenamefont {Sanpera}}]{Guhne2002}%
  \BibitemOpen
  \bibfield  {author} {\bibinfo {author} {\bibfnamefont {O.}~\bibnamefont
  {G\"uhne}}, \bibinfo {author} {\bibfnamefont {P.}~\bibnamefont {Hyllus}},
  \bibinfo {author} {\bibfnamefont {D.}~\bibnamefont {Bru\ss{}}}, \bibinfo
  {author} {\bibfnamefont {A.}~\bibnamefont {Ekert}}, \bibinfo {author}
  {\bibfnamefont {M.}~\bibnamefont {Lewenstein}}, \bibinfo {author}
  {\bibfnamefont {C.}~\bibnamefont {Macchiavello}},\ and\ \bibinfo {author}
  {\bibfnamefont {A.}~\bibnamefont {Sanpera}},\ }\bibfield  {title} {\bibinfo
  {title} {Detection of entanglement with few local measurements},\ }\href
  {https://doi.org/10.1103/PhysRevA.66.062305} {\bibfield  {journal} {\bibinfo
  {journal} {Phys. Rev. A}\ }\textbf {\bibinfo {volume} {66}},\ \bibinfo
  {pages} {062305} (\bibinfo {year} {2002})}\BibitemShut {NoStop}%
\bibitem [{\citenamefont {Bourennane}\ \emph {et~al.}(2004)\citenamefont
  {Bourennane}, \citenamefont {Eibl}, \citenamefont {Kurtsiefer}, \citenamefont
  {Gaertner}, \citenamefont {Weinfurter}, \citenamefont {G\"uhne},
  \citenamefont {Hyllus}, \citenamefont {Bru\ss{}}, \citenamefont
  {Lewenstein},\ and\ \citenamefont {Sanpera}}]{Mohamed2004}%
  \BibitemOpen
  \bibfield  {author} {\bibinfo {author} {\bibfnamefont {M.}~\bibnamefont
  {Bourennane}}, \bibinfo {author} {\bibfnamefont {M.}~\bibnamefont {Eibl}},
  \bibinfo {author} {\bibfnamefont {C.}~\bibnamefont {Kurtsiefer}}, \bibinfo
  {author} {\bibfnamefont {S.}~\bibnamefont {Gaertner}}, \bibinfo {author}
  {\bibfnamefont {H.}~\bibnamefont {Weinfurter}}, \bibinfo {author}
  {\bibfnamefont {O.}~\bibnamefont {G\"uhne}}, \bibinfo {author} {\bibfnamefont
  {P.}~\bibnamefont {Hyllus}}, \bibinfo {author} {\bibfnamefont
  {D.}~\bibnamefont {Bru\ss{}}}, \bibinfo {author} {\bibfnamefont
  {M.}~\bibnamefont {Lewenstein}},\ and\ \bibinfo {author} {\bibfnamefont
  {A.}~\bibnamefont {Sanpera}},\ }\bibfield  {title} {\bibinfo {title}
  {Experimental detection of multipartite entanglement using witness
  operators},\ }\href {https://doi.org/10.1103/PhysRevLett.92.087902}
  {\bibfield  {journal} {\bibinfo  {journal} {Phys. Rev. Lett.}\ }\textbf
  {\bibinfo {volume} {92}},\ \bibinfo {pages} {087902} (\bibinfo {year}
  {2004})}\BibitemShut {NoStop}%
\bibitem [{\citenamefont {Barbieri}\ \emph {et~al.}(2003)\citenamefont
  {Barbieri}, \citenamefont {De~Martini}, \citenamefont {Di~Nepi},
  \citenamefont {Mataloni}, \citenamefont {D'Ariano},\ and\ \citenamefont
  {Macchiavello}}]{Barbieri2003}%
  \BibitemOpen
  \bibfield  {author} {\bibinfo {author} {\bibfnamefont {M.}~\bibnamefont
  {Barbieri}}, \bibinfo {author} {\bibfnamefont {F.}~\bibnamefont
  {De~Martini}}, \bibinfo {author} {\bibfnamefont {G.}~\bibnamefont {Di~Nepi}},
  \bibinfo {author} {\bibfnamefont {P.}~\bibnamefont {Mataloni}}, \bibinfo
  {author} {\bibfnamefont {G.~M.}\ \bibnamefont {D'Ariano}},\ and\ \bibinfo
  {author} {\bibfnamefont {C.}~\bibnamefont {Macchiavello}},\ }\bibfield
  {title} {\bibinfo {title} {Detection of entanglement with polarized photons:
  Experimental realization of an entanglement witness},\ }\href
  {https://doi.org/10.1103/PhysRevLett.91.227901} {\bibfield  {journal}
  {\bibinfo  {journal} {Phys. Rev. Lett.}\ }\textbf {\bibinfo {volume} {91}},\
  \bibinfo {pages} {227901} (\bibinfo {year} {2003})}\BibitemShut {NoStop}%
\bibitem [{\citenamefont {Ganguly}\ \emph {et~al.}(2011)\citenamefont
  {Ganguly}, \citenamefont {Adhikari}, \citenamefont {Majumdar},\ and\
  \citenamefont {Chatterjee}}]{Ganguly2011}%
  \BibitemOpen
  \bibfield  {author} {\bibinfo {author} {\bibfnamefont {N.}~\bibnamefont
  {Ganguly}}, \bibinfo {author} {\bibfnamefont {S.}~\bibnamefont {Adhikari}},
  \bibinfo {author} {\bibfnamefont {A.~S.}\ \bibnamefont {Majumdar}},\ and\
  \bibinfo {author} {\bibfnamefont {J.}~\bibnamefont {Chatterjee}},\ }\bibfield
   {title} {\bibinfo {title} {Entanglement witness operator for quantum
  teleportation},\ }\href {https://doi.org/10.1103/PhysRevLett.107.270501}
  {\bibfield  {journal} {\bibinfo  {journal} {Phys. Rev. Lett.}\ }\textbf
  {\bibinfo {volume} {107}},\ \bibinfo {pages} {270501} (\bibinfo {year}
  {2011})}\BibitemShut {NoStop}%
\bibitem [{\citenamefont {Hyllus}(2005)}]{Hyllus}%
  \BibitemOpen
  \bibfield  {author} {\bibinfo {author} {\bibfnamefont {P.}~\bibnamefont
  {Hyllus}},\ }\href@noop {} {\emph {\bibinfo {title} {Witnessing entanglement
  in qudit systems}}}\ (\bibinfo  {publisher} {PhD Thesis, University of
  Hannover},\ \bibinfo {year} {2005})\BibitemShut {NoStop}%
\bibitem [{\citenamefont {Coffman}\ \emph
  {et~al.}(2000{\natexlab{b}})\citenamefont {Coffman}, \citenamefont {Kundu},\
  and\ \citenamefont {Wootters}}]{CKW2000}%
  \BibitemOpen
  \bibfield  {author} {\bibinfo {author} {\bibfnamefont {V.}~\bibnamefont
  {Coffman}}, \bibinfo {author} {\bibfnamefont {J.}~\bibnamefont {Kundu}},\
  and\ \bibinfo {author} {\bibfnamefont {W.~K.}\ \bibnamefont {Wootters}},\
  }\bibfield  {title} {\bibinfo {title} {Distributed entanglement},\ }\href
  {https://doi.org/10.1103/PhysRevA.61.052306} {\bibfield  {journal} {\bibinfo
  {journal} {Phys. Rev. A}\ }\textbf {\bibinfo {volume} {61}},\ \bibinfo
  {pages} {052306} (\bibinfo {year} {2000}{\natexlab{b}})}\BibitemShut
  {NoStop}%
\bibitem [{\citenamefont {Toner}\ and\ \citenamefont
  {Verstraete}(2006)}]{Toner2006}%
  \BibitemOpen
  \bibfield  {author} {\bibinfo {author} {\bibfnamefont {B.}~\bibnamefont
  {Toner}}\ and\ \bibinfo {author} {\bibfnamefont {F.}~\bibnamefont
  {Verstraete}},\ }\href {https://arxiv.org/abs/quant-ph/0611001} {\bibinfo
  {title} {Monogamy of bell correlations and tsirelson's bound}} (\bibinfo
  {year} {2006}),\ \Eprint {https://arxiv.org/abs/quant-ph/0611001}
  {arXiv:quant-ph/0611001 [quant-ph]} \BibitemShut {NoStop}%
\bibitem [{\citenamefont {Toner}(2008)}]{Toner2008}%
  \BibitemOpen
  \bibfield  {author} {\bibinfo {author} {\bibfnamefont {B.}~\bibnamefont
  {Toner}},\ }\bibfield  {title} {\bibinfo {title} {Monogamy of non-local
  quantum correlations},\ }\href
  {https://doi.org/https://doi.org/10.1098/rspa.2008.0149} {\bibfield
  {journal} {\bibinfo  {journal} {Proc. R. Soc. A.}\ }\textbf {\bibinfo
  {volume} {465}},\ \bibinfo {pages} {59} (\bibinfo {year} {2008})}\BibitemShut
  {NoStop}%
\bibitem [{\citenamefont {Kurzy\ifmmode~\acute{n}\else \'{n}\fi{}ski}\ \emph
  {et~al.}(2011)\citenamefont {Kurzy\ifmmode~\acute{n}\else \'{n}\fi{}ski},
  \citenamefont {Paterek}, \citenamefont {Ramanathan}, \citenamefont
  {Laskowski},\ and\ \citenamefont {Kaszlikowski}}]{Kurz2011}%
  \BibitemOpen
  \bibfield  {author} {\bibinfo {author} {\bibfnamefont {P.}~\bibnamefont
  {Kurzy\ifmmode~\acute{n}\else \'{n}\fi{}ski}}, \bibinfo {author}
  {\bibfnamefont {T.}~\bibnamefont {Paterek}}, \bibinfo {author} {\bibfnamefont
  {R.}~\bibnamefont {Ramanathan}}, \bibinfo {author} {\bibfnamefont
  {W.}~\bibnamefont {Laskowski}},\ and\ \bibinfo {author} {\bibfnamefont
  {D.}~\bibnamefont {Kaszlikowski}},\ }\bibfield  {title} {\bibinfo {title}
  {Correlation complementarity yields bell monogamy relations},\ }\href
  {https://doi.org/10.1103/PhysRevLett.106.180402} {\bibfield  {journal}
  {\bibinfo  {journal} {Phys. Rev. Lett.}\ }\textbf {\bibinfo {volume} {106}},\
  \bibinfo {pages} {180402} (\bibinfo {year} {2011})}\BibitemShut {NoStop}%
\bibitem [{\citenamefont {Reid}(2013)}]{Reid2013}%
  \BibitemOpen
  \bibfield  {author} {\bibinfo {author} {\bibfnamefont {M.~D.}\ \bibnamefont
  {Reid}},\ }\bibfield  {title} {\bibinfo {title} {Monogamy inequalities for
  the einstein-podolsky-rosen paradox and quantum steering},\ }\href
  {https://doi.org/10.1103/PhysRevA.88.062108} {\bibfield  {journal} {\bibinfo
  {journal} {Phys. Rev. A}\ }\textbf {\bibinfo {volume} {88}},\ \bibinfo
  {pages} {062108} (\bibinfo {year} {2013})}\BibitemShut {NoStop}%
\bibitem [{\citenamefont {Milne}\ \emph {et~al.}(2014)\citenamefont {Milne},
  \citenamefont {Jevtic}, \citenamefont {Jennings}, \citenamefont {Wiseman},\
  and\ \citenamefont {Rudolph}}]{Milne2014}%
  \BibitemOpen
  \bibfield  {author} {\bibinfo {author} {\bibfnamefont {A.}~\bibnamefont
  {Milne}}, \bibinfo {author} {\bibfnamefont {S.}~\bibnamefont {Jevtic}},
  \bibinfo {author} {\bibfnamefont {D.}~\bibnamefont {Jennings}}, \bibinfo
  {author} {\bibfnamefont {H.}~\bibnamefont {Wiseman}},\ and\ \bibinfo {author}
  {\bibfnamefont {T.}~\bibnamefont {Rudolph}},\ }\bibfield  {title} {\bibinfo
  {title} {Quantum steering ellipsoids, extremal physical states and
  monogamy},\ }\href {https://doi.org/10.1088/1367-2630/16/8/083017} {\bibfield
   {journal} {\bibinfo  {journal} {New J. Phys.}\ }\textbf {\bibinfo {volume}
  {16}},\ \bibinfo {pages} {083017} (\bibinfo {year} {2014})}\BibitemShut
  {NoStop}%
\bibitem [{\citenamefont {Lee}\ and\ \citenamefont {Park}(2009)}]{Lee2009}%
  \BibitemOpen
  \bibfield  {author} {\bibinfo {author} {\bibfnamefont {S.}~\bibnamefont
  {Lee}}\ and\ \bibinfo {author} {\bibfnamefont {J.}~\bibnamefont {Park}},\
  }\bibfield  {title} {\bibinfo {title} {Monogamy of entanglement and
  teleportation capability},\ }\href
  {https://doi.org/10.1103/PhysRevA.79.054309} {\bibfield  {journal} {\bibinfo
  {journal} {Phys. Rev. A}\ }\textbf {\bibinfo {volume} {79}},\ \bibinfo
  {pages} {054309} (\bibinfo {year} {2009})}\BibitemShut {NoStop}%
\bibitem [{\citenamefont {Ramanathan}\ \emph {et~al.}(2012)\citenamefont
  {Ramanathan}, \citenamefont {Soeda}, \citenamefont
  {Kurzy\ifmmode~\acute{n}\else \'{n}\fi{}ski},\ and\ \citenamefont
  {Kaszlikowski}}]{Ramanathan2012}%
  \BibitemOpen
  \bibfield  {author} {\bibinfo {author} {\bibfnamefont {R.}~\bibnamefont
  {Ramanathan}}, \bibinfo {author} {\bibfnamefont {A.}~\bibnamefont {Soeda}},
  \bibinfo {author} {\bibfnamefont {P.}~\bibnamefont
  {Kurzy\ifmmode~\acute{n}\else \'{n}\fi{}ski}},\ and\ \bibinfo {author}
  {\bibfnamefont {D.}~\bibnamefont {Kaszlikowski}},\ }\bibfield  {title}
  {\bibinfo {title} {Generalized monogamy of contextual inequalities from the
  no-disturbance principle},\ }\href
  {https://doi.org/10.1103/PhysRevLett.109.050404} {\bibfield  {journal}
  {\bibinfo  {journal} {Phys. Rev. Lett.}\ }\textbf {\bibinfo {volume} {109}},\
  \bibinfo {pages} {050404} (\bibinfo {year} {2012})}\BibitemShut {NoStop}%
\bibitem [{\citenamefont {Kurzy\ifmmode~\acute{n}\else \'{n}\fi{}ski}\ \emph
  {et~al.}(2014)\citenamefont {Kurzy\ifmmode~\acute{n}\else \'{n}\fi{}ski},
  \citenamefont {Cabello},\ and\ \citenamefont {Kaszlikowski}}]{Kurz2014}%
  \BibitemOpen
  \bibfield  {author} {\bibinfo {author} {\bibfnamefont {P.}~\bibnamefont
  {Kurzy\ifmmode~\acute{n}\else \'{n}\fi{}ski}}, \bibinfo {author}
  {\bibfnamefont {A.}~\bibnamefont {Cabello}},\ and\ \bibinfo {author}
  {\bibfnamefont {D.}~\bibnamefont {Kaszlikowski}},\ }\bibfield  {title}
  {\bibinfo {title} {Fundamental monogamy relation between contextuality and
  nonlocality},\ }\href {https://doi.org/10.1103/PhysRevLett.112.100401}
  {\bibfield  {journal} {\bibinfo  {journal} {Phys. Rev. Lett.}\ }\textbf
  {\bibinfo {volume} {112}},\ \bibinfo {pages} {100401} (\bibinfo {year}
  {2014})}\BibitemShut {NoStop}%
\bibitem [{\citenamefont {Giorgi}(2011)}]{Giorgi2011}%
  \BibitemOpen
  \bibfield  {author} {\bibinfo {author} {\bibfnamefont {G.~L.}\ \bibnamefont
  {Giorgi}},\ }\bibfield  {title} {\bibinfo {title} {Monogamy properties of
  quantum and classical correlations},\ }\href
  {https://doi.org/10.1103/PhysRevA.84.054301} {\bibfield  {journal} {\bibinfo
  {journal} {Phys. Rev. A}\ }\textbf {\bibinfo {volume} {84}},\ \bibinfo
  {pages} {054301} (\bibinfo {year} {2011})}\BibitemShut {NoStop}%
\bibitem [{\citenamefont {Allegra}\ \emph {et~al.}(2011)\citenamefont
  {Allegra}, \citenamefont {Giorda},\ and\ \citenamefont
  {Montorsi}}]{Allegra2011}%
  \BibitemOpen
  \bibfield  {author} {\bibinfo {author} {\bibfnamefont {M.}~\bibnamefont
  {Allegra}}, \bibinfo {author} {\bibfnamefont {P.}~\bibnamefont {Giorda}},\
  and\ \bibinfo {author} {\bibfnamefont {A.}~\bibnamefont {Montorsi}},\
  }\bibfield  {title} {\bibinfo {title} {Quantum discord and classical
  correlations in the bond-charge hubbard model: Quantum phase transitions,
  off-diagonal long-range order, and violation of the monogamy property for
  discord},\ }\href {https://doi.org/10.1103/PhysRevB.84.245133} {\bibfield
  {journal} {\bibinfo  {journal} {Phys. Rev. B}\ }\textbf {\bibinfo {volume}
  {84}},\ \bibinfo {pages} {245133} (\bibinfo {year} {2011})}\BibitemShut
  {NoStop}%
\bibitem [{\citenamefont {Song}\ \emph {et~al.}(2013)\citenamefont {Song},
  \citenamefont {Wu},\ and\ \citenamefont {Ye}}]{Song2013}%
  \BibitemOpen
  \bibfield  {author} {\bibinfo {author} {\bibfnamefont {X.-K.}\ \bibnamefont
  {Song}}, \bibinfo {author} {\bibfnamefont {T.}~\bibnamefont {Wu}},\ and\
  \bibinfo {author} {\bibfnamefont {L.}~\bibnamefont {Ye}},\ }\bibfield
  {title} {\bibinfo {title} {The monogamy relation and quantum phase transition
  in one-dimensional anisotropic xxz model},\ }\href
  {https://doi.org/10.1007/s11128-013-0598-5} {\ \textbf {\bibinfo {volume}
  {12}},\ \bibinfo {pages} {3305–3317} (\bibinfo {year} {2013})}\BibitemShut
  {NoStop}%
\bibitem [{\citenamefont {Kumar}\ \emph {et~al.}(2016)\citenamefont {Kumar},
  \citenamefont {{Singha Roy}}, \citenamefont {Pal}, \citenamefont {Prabhu},
  \citenamefont {Sen(De)},\ and\ \citenamefont {Sen}}]{Kumar2016}%
  \BibitemOpen
  \bibfield  {author} {\bibinfo {author} {\bibfnamefont {A.}~\bibnamefont
  {Kumar}}, \bibinfo {author} {\bibfnamefont {S.}~\bibnamefont {{Singha Roy}}},
  \bibinfo {author} {\bibfnamefont {A.~K.}\ \bibnamefont {Pal}}, \bibinfo
  {author} {\bibfnamefont {R.}~\bibnamefont {Prabhu}}, \bibinfo {author}
  {\bibfnamefont {A.}~\bibnamefont {Sen(De)}},\ and\ \bibinfo {author}
  {\bibfnamefont {U.}~\bibnamefont {Sen}},\ }\bibfield  {title} {\bibinfo
  {title} {Conclusive identification of quantum channels via monogamy of
  quantum correlations},\ }\href
  {https://doi.org/https://doi.org/10.1016/j.physleta.2016.08.039} {\bibfield
  {journal} {\bibinfo  {journal} {Phys. Lett. A}\ }\textbf {\bibinfo {volume}
  {380}},\ \bibinfo {pages} {3588} (\bibinfo {year} {2016})}\BibitemShut
  {NoStop}%
\bibitem [{\citenamefont {Gisin}\ \emph {et~al.}(2002)\citenamefont {Gisin},
  \citenamefont {Ribordy}, \citenamefont {Tittel},\ and\ \citenamefont
  {Zbinden}}]{Gisin2002}%
  \BibitemOpen
  \bibfield  {author} {\bibinfo {author} {\bibfnamefont {N.}~\bibnamefont
  {Gisin}}, \bibinfo {author} {\bibfnamefont {G.}~\bibnamefont {Ribordy}},
  \bibinfo {author} {\bibfnamefont {W.}~\bibnamefont {Tittel}},\ and\ \bibinfo
  {author} {\bibfnamefont {H.}~\bibnamefont {Zbinden}},\ }\bibfield  {title}
  {\bibinfo {title} {Quantum cryptography},\ }\href
  {https://doi.org/10.1103/RevModPhys.74.145} {\bibfield  {journal} {\bibinfo
  {journal} {Rev. Mod. Phys.}\ }\textbf {\bibinfo {volume} {74}},\ \bibinfo
  {pages} {145} (\bibinfo {year} {2002})}\BibitemShut {NoStop}%
\bibitem [{\citenamefont {Pramanik}\ \emph
  {et~al.}(2014{\natexlab{b}})\citenamefont {Pramanik}, \citenamefont
  {Kaplan},\ and\ \citenamefont {Majumdar}}]{TP2014}%
  \BibitemOpen
  \bibfield  {author} {\bibinfo {author} {\bibfnamefont {T.}~\bibnamefont
  {Pramanik}}, \bibinfo {author} {\bibfnamefont {M.}~\bibnamefont {Kaplan}},\
  and\ \bibinfo {author} {\bibfnamefont {A.~S.}\ \bibnamefont {Majumdar}},\
  }\bibfield  {title} {\bibinfo {title} {Fine-grained
  einstein-podolsky-rosen--steering inequalities},\ }\href
  {https://doi.org/10.1103/PhysRevA.90.050305} {\bibfield  {journal} {\bibinfo
  {journal} {Phys. Rev. A}\ }\textbf {\bibinfo {volume} {90}},\ \bibinfo
  {pages} {050305} (\bibinfo {year} {2014}{\natexlab{b}})}\BibitemShut
  {NoStop}%
\bibitem [{\citenamefont {Datta}\ \emph {et~al.}(2017)\citenamefont {Datta},
  \citenamefont {Goswami}, \citenamefont {Pramanik},\ and\ \citenamefont
  {Majumdar}}]{Datta2017}%
  \BibitemOpen
  \bibfield  {author} {\bibinfo {author} {\bibfnamefont {S.}~\bibnamefont
  {Datta}}, \bibinfo {author} {\bibfnamefont {S.}~\bibnamefont {Goswami}},
  \bibinfo {author} {\bibfnamefont {T.}~\bibnamefont {Pramanik}},\ and\
  \bibinfo {author} {\bibfnamefont {A.}~\bibnamefont {Majumdar}},\ }\bibfield
  {title} {\bibinfo {title} {Preservation of a lower bound of quantum secret
  key rate in the presence of decoherence},\ }\href
  {https://doi.org/https://doi.org/10.1016/j.physleta.2017.01.019} {\bibfield
  {journal} {\bibinfo  {journal} {Phys. Lett. A}\ }\textbf {\bibinfo {volume}
  {381}},\ \bibinfo {pages} {897} (\bibinfo {year} {2017})}\BibitemShut
  {NoStop}%
\bibitem [{\citenamefont {Greenberger}\ \emph {et~al.}(1990)\citenamefont
  {Greenberger}, \citenamefont {Horne}, \citenamefont {Shimony},\ and\
  \citenamefont {Zeilinger}}]{GHZ1990}%
  \BibitemOpen
  \bibfield  {author} {\bibinfo {author} {\bibfnamefont {D.~M.}\ \bibnamefont
  {Greenberger}}, \bibinfo {author} {\bibfnamefont {M.~A.}\ \bibnamefont
  {Horne}}, \bibinfo {author} {\bibfnamefont {A.}~\bibnamefont {Shimony}},\
  and\ \bibinfo {author} {\bibfnamefont {A.}~\bibnamefont {Zeilinger}},\
  }\bibfield  {title} {\bibinfo {title} {Bell’s theorem without
  inequalities},\ }\href {https://doi.org/https://doi.org/10.1119/1.16243}
  {\bibfield  {journal} {\bibinfo  {journal} {Am. J. Phys.}\ }\textbf {\bibinfo
  {volume} {58}},\ \bibinfo {pages} {1131–1143} (\bibinfo {year}
  {1990})}\BibitemShut {NoStop}%
\bibitem [{\citenamefont {D\"ur}\ \emph {et~al.}(2000)\citenamefont {D\"ur},
  \citenamefont {Vidal},\ and\ \citenamefont {Cirac}}]{Dur2000}%
  \BibitemOpen
  \bibfield  {author} {\bibinfo {author} {\bibfnamefont {W.}~\bibnamefont
  {D\"ur}}, \bibinfo {author} {\bibfnamefont {G.}~\bibnamefont {Vidal}},\ and\
  \bibinfo {author} {\bibfnamefont {J.~I.}\ \bibnamefont {Cirac}},\ }\bibfield
  {title} {\bibinfo {title} {Three qubits can be entangled in two inequivalent
  ways},\ }\href {https://doi.org/10.1103/PhysRevA.62.062314} {\bibfield
  {journal} {\bibinfo  {journal} {Phys. Rev. A}\ }\textbf {\bibinfo {volume}
  {62}},\ \bibinfo {pages} {062314} (\bibinfo {year} {2000})}\BibitemShut
  {NoStop}%
\bibitem [{\citenamefont {Ac\'{\i}n}\ \emph {et~al.}(2001)\citenamefont
  {Ac\'{\i}n}, \citenamefont {Bru\ss{}}, \citenamefont {Lewenstein},\ and\
  \citenamefont {Sanpera}}]{Acin2001}%
  \BibitemOpen
  \bibfield  {author} {\bibinfo {author} {\bibfnamefont {A.}~\bibnamefont
  {Ac\'{\i}n}}, \bibinfo {author} {\bibfnamefont {D.}~\bibnamefont {Bru\ss{}}},
  \bibinfo {author} {\bibfnamefont {M.}~\bibnamefont {Lewenstein}},\ and\
  \bibinfo {author} {\bibfnamefont {A.}~\bibnamefont {Sanpera}},\ }\bibfield
  {title} {\bibinfo {title} {Classification of mixed three-qubit states},\
  }\href {https://doi.org/10.1103/PhysRevLett.87.040401} {\bibfield  {journal}
  {\bibinfo  {journal} {Phys. Rev. Lett.}\ }\textbf {\bibinfo {volume} {87}},\
  \bibinfo {pages} {040401} (\bibinfo {year} {2001})}\BibitemShut {NoStop}%
\bibitem [{\citenamefont {Neumann}(1955)}]{Neumann}%
  \BibitemOpen
  \bibfield  {author} {\bibinfo {author} {\bibfnamefont {J.~V.}\ \bibnamefont
  {Neumann}},\ }\href@noop {} {\emph {\bibinfo {title} {Mathematical
  Foundations of Quantum Mechanics}}}\ (\bibinfo  {publisher} {Princeton
  University Press},\ \bibinfo {year} {1955})\BibitemShut {NoStop}%
\bibitem [{\citenamefont {Xue}\ \emph {et~al.}(2021)\citenamefont {Xue},
  \citenamefont {Guo}, \citenamefont {Li}, \citenamefont {Ye},\ and\
  \citenamefont {Li}}]{Xue2021}%
  \BibitemOpen
  \bibfield  {author} {\bibinfo {author} {\bibfnamefont {S.}~\bibnamefont
  {Xue}}, \bibinfo {author} {\bibfnamefont {J.}~\bibnamefont {Guo}}, \bibinfo
  {author} {\bibfnamefont {P.}~\bibnamefont {Li}}, \bibinfo {author}
  {\bibfnamefont {M.}~\bibnamefont {Ye}},\ and\ \bibinfo {author}
  {\bibfnamefont {Y.}~\bibnamefont {Li}},\ }\bibfield  {title} {\bibinfo
  {title} {Quantification of resource theory of imaginarity},\ }\href
  {https://doi.org/https://doi.org/10.1007/s11128-021-03324-5} {\bibfield
  {journal} {\bibinfo  {journal} {Quant. Inf. Process.}\ }\textbf {\bibinfo
  {volume} {20}},\ \bibinfo {pages} {383} (\bibinfo {year} {2021})}\BibitemShut
  {NoStop}%
\bibitem [{\citenamefont {Chen}\ \emph {et~al.}(2023)\citenamefont {Chen},
  \citenamefont {Gao},\ and\ \citenamefont {Yan}}]{Chen2023}%
  \BibitemOpen
  \bibfield  {author} {\bibinfo {author} {\bibfnamefont {Q.}~\bibnamefont
  {Chen}}, \bibinfo {author} {\bibfnamefont {T.}~\bibnamefont {Gao}},\ and\
  \bibinfo {author} {\bibfnamefont {F.}~\bibnamefont {Yan}},\ }\bibfield
  {title} {\bibinfo {title} {Measures of imaginarity and quantum state order},\
  }\href {https://doi.org/https://doi.org/10.1007/s11433-023-2126-9} {\bibfield
   {journal} {\bibinfo  {journal} {Sci. China Phys. Mech. Astron.}\ }\textbf
  {\bibinfo {volume} {66}},\ \bibinfo {pages} {280312} (\bibinfo {year}
  {2023})}\BibitemShut {NoStop}%
\bibitem [{\citenamefont {Nielsen}\ and\ \citenamefont
  {Chuang}(2010)}]{NC2010}%
  \BibitemOpen
  \bibfield  {author} {\bibinfo {author} {\bibfnamefont {M.~A.}\ \bibnamefont
  {Nielsen}}\ and\ \bibinfo {author} {\bibfnamefont {I.~L.}\ \bibnamefont
  {Chuang}},\ }\href@noop {} {\emph {\bibinfo {title} {Quantum Computation and
  Quantum Information: 10th Anniversary Edition}}}\ (\bibinfo  {publisher}
  {Cambridge University Press},\ \bibinfo {year} {2010})\BibitemShut {NoStop}%
\bibitem [{\citenamefont {Brunner}\ \emph {et~al.}(2014)\citenamefont
  {Brunner}, \citenamefont {Cavalcanti}, \citenamefont {Pironio}, \citenamefont
  {Scarani},\ and\ \citenamefont {Wehner}}]{Brunner2014}%
  \BibitemOpen
  \bibfield  {author} {\bibinfo {author} {\bibfnamefont {N.}~\bibnamefont
  {Brunner}}, \bibinfo {author} {\bibfnamefont {D.}~\bibnamefont {Cavalcanti}},
  \bibinfo {author} {\bibfnamefont {S.}~\bibnamefont {Pironio}}, \bibinfo
  {author} {\bibfnamefont {V.}~\bibnamefont {Scarani}},\ and\ \bibinfo {author}
  {\bibfnamefont {S.}~\bibnamefont {Wehner}},\ }\bibfield  {title} {\bibinfo
  {title} {Bell nonlocality},\ }\href
  {https://doi.org/10.1103/RevModPhys.86.419} {\bibfield  {journal} {\bibinfo
  {journal} {Rev. Mod. Phys.}\ }\textbf {\bibinfo {volume} {86}},\ \bibinfo
  {pages} {419} (\bibinfo {year} {2014})}\BibitemShut {NoStop}%
\bibitem [{\citenamefont {Cavalcanti}\ and\ \citenamefont
  {Skrzypczyk}(2016)}]{Cavalcanti2017}%
  \BibitemOpen
  \bibfield  {author} {\bibinfo {author} {\bibfnamefont {D.}~\bibnamefont
  {Cavalcanti}}\ and\ \bibinfo {author} {\bibfnamefont {P.}~\bibnamefont
  {Skrzypczyk}},\ }\bibfield  {title} {\bibinfo {title} {Quantum steering: a
  review with focus on semidefinite programming},\ }\href
  {https://doi.org/10.1088/1361-6633/80/2/024001} {\bibfield  {journal}
  {\bibinfo  {journal} {Rep. Prog. Phys.}\ }\textbf {\bibinfo {volume} {80}},\
  \bibinfo {pages} {024001} (\bibinfo {year} {2016})}\BibitemShut {NoStop}%
\bibitem [{\citenamefont {Busch}\ \emph {et~al.}(1996)\citenamefont {Busch},
  \citenamefont {Lahti},\ and\ \citenamefont {Mittelstaedt}}]{Busch}%
  \BibitemOpen
  \bibfield  {author} {\bibinfo {author} {\bibfnamefont {P.}~\bibnamefont
  {Busch}}, \bibinfo {author} {\bibfnamefont {P.~J.}\ \bibnamefont {Lahti}},\
  and\ \bibinfo {author} {\bibfnamefont {P.}~\bibnamefont {Mittelstaedt}},\
  }\href@noop {} {\emph {\bibinfo {title} {The quantum theory of
  measurement}}},\ Vol.~\bibinfo {volume} {2}\ (\bibinfo  {publisher} {Springer
  Science \& Business Media: Berlin, Germany},\ \bibinfo {year}
  {1996})\BibitemShut {NoStop}%
\bibitem [{\citenamefont {Mal}\ \emph {et~al.}(2016)\citenamefont {Mal},
  \citenamefont {Majumdar},\ and\ \citenamefont {Home}}]{Mal1}%
  \BibitemOpen
  \bibfield  {author} {\bibinfo {author} {\bibfnamefont {S.}~\bibnamefont
  {Mal}}, \bibinfo {author} {\bibfnamefont {A.~S.}\ \bibnamefont {Majumdar}},\
  and\ \bibinfo {author} {\bibfnamefont {D.}~\bibnamefont {Home}},\ }\bibfield
  {title} {\bibinfo {title} {Sharing of nonlocality of a single member of an
  entangled pair of qubits is not possible by more than two unbiased observers
  on the other wing},\ }\href {https://doi.org/10.3390/math4030048} {\bibfield
  {journal} {\bibinfo  {journal} {Mathematics}\ }\textbf {\bibinfo {volume}
  {4}},\ \bibinfo {pages} {48} (\bibinfo {year} {2016})}\BibitemShut {NoStop}%
\bibitem [{\citenamefont {Brand\~ao}(2005)}]{Brandao2005}%
  \BibitemOpen
  \bibfield  {author} {\bibinfo {author} {\bibfnamefont {F.~G. S.~L.}\
  \bibnamefont {Brand\~ao}},\ }\bibfield  {title} {\bibinfo {title}
  {Quantifying entanglement with witness operators},\ }\href
  {https://doi.org/10.1103/PhysRevA.72.022310} {\bibfield  {journal} {\bibinfo
  {journal} {Phys. Rev. A}\ }\textbf {\bibinfo {volume} {72}},\ \bibinfo
  {pages} {022310} (\bibinfo {year} {2005})}\BibitemShut {NoStop}%
\bibitem [{\citenamefont {Cai}\ \emph {et~al.}(2023)\citenamefont {Cai},
  \citenamefont {Babbush}, \citenamefont {Benjamin}, \citenamefont {Endo},
  \citenamefont {Huggins}, \citenamefont {Li}, \citenamefont {McClean},\ and\
  \citenamefont {O'Brien}}]{Cai2023}%
  \BibitemOpen
  \bibfield  {author} {\bibinfo {author} {\bibfnamefont {Z.}~\bibnamefont
  {Cai}}, \bibinfo {author} {\bibfnamefont {R.}~\bibnamefont {Babbush}},
  \bibinfo {author} {\bibfnamefont {S.~C.}\ \bibnamefont {Benjamin}}, \bibinfo
  {author} {\bibfnamefont {S.}~\bibnamefont {Endo}}, \bibinfo {author}
  {\bibfnamefont {W.~J.}\ \bibnamefont {Huggins}}, \bibinfo {author}
  {\bibfnamefont {Y.}~\bibnamefont {Li}}, \bibinfo {author} {\bibfnamefont
  {J.~R.}\ \bibnamefont {McClean}},\ and\ \bibinfo {author} {\bibfnamefont
  {T.~E.}\ \bibnamefont {O'Brien}},\ }\bibfield  {title} {\bibinfo {title}
  {Quantum error mitigation},\ }\href
  {https://doi.org/10.1103/RevModPhys.95.045005} {\bibfield  {journal}
  {\bibinfo  {journal} {Rev. Mod. Phys.}\ }\textbf {\bibinfo {volume} {95}},\
  \bibinfo {pages} {045005} (\bibinfo {year} {2023})}\BibitemShut {NoStop}%
\bibitem [{\citenamefont {Passaro}\ \emph {et~al.}(2015)\citenamefont
  {Passaro}, \citenamefont {Cavalcanti}, \citenamefont {Skrzypczyk},\ and\
  \citenamefont {Ac{\'{\i}}n}}]{Passaro}%
  \BibitemOpen
  \bibfield  {author} {\bibinfo {author} {\bibfnamefont {E.}~\bibnamefont
  {Passaro}}, \bibinfo {author} {\bibfnamefont {D.}~\bibnamefont {Cavalcanti}},
  \bibinfo {author} {\bibfnamefont {P.}~\bibnamefont {Skrzypczyk}},\ and\
  \bibinfo {author} {\bibfnamefont {A.}~\bibnamefont {Ac{\'{\i}}n}},\
  }\bibfield  {title} {\bibinfo {title} {Optimal randomness certification in
  the quantum steering and prepare-and-measure scenarios},\ }\href
  {https://doi.org/10.1088/1367-2630/17/11/113010} {\bibfield  {journal}
  {\bibinfo  {journal} {New J. Phys.}\ }\textbf {\bibinfo {volume} {17}},\
  \bibinfo {pages} {113010} (\bibinfo {year} {2015})}\BibitemShut {NoStop}%
\bibitem [{\citenamefont {\v{S}upic}\ and\ \citenamefont
  {Hoban}(2016)}]{Supic}%
  \BibitemOpen
  \bibfield  {author} {\bibinfo {author} {\bibfnamefont {I.}~\bibnamefont
  {\v{S}upic}}\ and\ \bibinfo {author} {\bibfnamefont {M.~J.}\ \bibnamefont
  {Hoban}},\ }\bibfield  {title} {\bibinfo {title} {Self-testing through
  epr-steering},\ }\href {https://doi.org/10.1088/1367-2630/18/7/075006}
  {\bibfield  {journal} {\bibinfo  {journal} {New J. of Phys.}\ }\textbf
  {\bibinfo {volume} {18}},\ \bibinfo {pages} {075006} (\bibinfo {year}
  {2016})}\BibitemShut {NoStop}%
\bibitem [{\citenamefont {Goswami}\ \emph {et~al.}(2018)\citenamefont
  {Goswami}, \citenamefont {Bhattacharya}, \citenamefont {Das}, \citenamefont
  {Sasmal}, \citenamefont {Jebaratnam},\ and\ \citenamefont
  {Majumdar}}]{Goswami}%
  \BibitemOpen
  \bibfield  {author} {\bibinfo {author} {\bibfnamefont {S.}~\bibnamefont
  {Goswami}}, \bibinfo {author} {\bibfnamefont {B.}~\bibnamefont
  {Bhattacharya}}, \bibinfo {author} {\bibfnamefont {D.}~\bibnamefont {Das}},
  \bibinfo {author} {\bibfnamefont {S.}~\bibnamefont {Sasmal}}, \bibinfo
  {author} {\bibfnamefont {C.}~\bibnamefont {Jebaratnam}},\ and\ \bibinfo
  {author} {\bibfnamefont {A.~S.}\ \bibnamefont {Majumdar}},\ }\bibfield
  {title} {\bibinfo {title} {One-sided device-independent self-testing of any
  pure two-qubit entangled state},\ }\href
  {https://doi.org/10.1103/PhysRevA.98.022311} {\bibfield  {journal} {\bibinfo
  {journal} {Phys. Rev. A}\ }\textbf {\bibinfo {volume} {98}},\ \bibinfo
  {pages} {022311} (\bibinfo {year} {2018})}\BibitemShut {NoStop}%
\bibitem [{\citenamefont {Mondal}\ and\ \citenamefont
  {Mukhopadhyay}(2015)}]{Mondal1}%
  \BibitemOpen
  \bibfield  {author} {\bibinfo {author} {\bibfnamefont {D.}~\bibnamefont
  {Mondal}}\ and\ \bibinfo {author} {\bibfnamefont {C.}~\bibnamefont
  {Mukhopadhyay}},\ }\href@noop {} {\bibinfo {title} {Steerability of quantum
  coherence in accelerated frame}} (\bibinfo {year} {2015}),\ \Eprint
  {https://arxiv.org/abs/1510.07556} {arXiv:1510.07556 [quant-ph]} \BibitemShut
  {NoStop}%
\bibitem [{\citenamefont {Hu}\ and\ \citenamefont {Fan}(2018)}]{Hu}%
  \BibitemOpen
  \bibfield  {author} {\bibinfo {author} {\bibfnamefont {M.-L.}\ \bibnamefont
  {Hu}}\ and\ \bibinfo {author} {\bibfnamefont {H.}~\bibnamefont {Fan}},\
  }\bibfield  {title} {\bibinfo {title} {Nonlocal advantage of quantum
  coherence in high-dimensional states},\ }\href
  {https://doi.org/10.1103/PhysRevA.98.022312} {\bibfield  {journal} {\bibinfo
  {journal} {Phys. Rev. A}\ }\textbf {\bibinfo {volume} {98}},\ \bibinfo
  {pages} {022312} (\bibinfo {year} {2018})}\BibitemShut {NoStop}%
\bibitem [{\citenamefont {Sasmal}\ \emph {et~al.}(2018)\citenamefont {Sasmal},
  \citenamefont {Das}, \citenamefont {Mal},\ and\ \citenamefont
  {Majumdar}}]{Sasmal2018}%
  \BibitemOpen
  \bibfield  {author} {\bibinfo {author} {\bibfnamefont {S.}~\bibnamefont
  {Sasmal}}, \bibinfo {author} {\bibfnamefont {D.}~\bibnamefont {Das}},
  \bibinfo {author} {\bibfnamefont {S.}~\bibnamefont {Mal}},\ and\ \bibinfo
  {author} {\bibfnamefont {A.~S.}\ \bibnamefont {Majumdar}},\ }\bibfield
  {title} {\bibinfo {title} {Steering a single system sequentially by multiple
  observers},\ }\href {https://doi.org/10.1103/PhysRevA.98.012305} {\bibfield
  {journal} {\bibinfo  {journal} {Phys. Rev. A}\ }\textbf {\bibinfo {volume}
  {98}},\ \bibinfo {pages} {012305} (\bibinfo {year} {2018})}\BibitemShut
  {NoStop}%
\bibitem [{\citenamefont {Datta}\ and\ \citenamefont
  {Majumdar}(2018)}]{Datta2018}%
  \BibitemOpen
  \bibfield  {author} {\bibinfo {author} {\bibfnamefont {S.}~\bibnamefont
  {Datta}}\ and\ \bibinfo {author} {\bibfnamefont {A.~S.}\ \bibnamefont
  {Majumdar}},\ }\bibfield  {title} {\bibinfo {title} {Sharing of nonlocal
  advantage of quantum coherence by sequential observers},\ }\href
  {https://doi.org/10.1103/PhysRevA.98.042311} {\bibfield  {journal} {\bibinfo
  {journal} {Phys. Rev. A}\ }\textbf {\bibinfo {volume} {98}},\ \bibinfo
  {pages} {042311} (\bibinfo {year} {2018})}\BibitemShut {NoStop}%
\bibitem [{\citenamefont {Gupta}\ \emph {et~al.}(2021)\citenamefont {Gupta},
  \citenamefont {Maity}, \citenamefont {Das}, \citenamefont {Roy},\ and\
  \citenamefont {Majumdar}}]{Gupta2021}%
  \BibitemOpen
  \bibfield  {author} {\bibinfo {author} {\bibfnamefont {S.}~\bibnamefont
  {Gupta}}, \bibinfo {author} {\bibfnamefont {A.~G.}\ \bibnamefont {Maity}},
  \bibinfo {author} {\bibfnamefont {D.}~\bibnamefont {Das}}, \bibinfo {author}
  {\bibfnamefont {A.}~\bibnamefont {Roy}},\ and\ \bibinfo {author}
  {\bibfnamefont {A.~S.}\ \bibnamefont {Majumdar}},\ }\bibfield  {title}
  {\bibinfo {title} {Genuine einstein-podolsky-rosen steering of three-qubit
  states by multiple sequential observers},\ }\href
  {https://doi.org/10.1103/PhysRevA.103.022421} {\bibfield  {journal} {\bibinfo
   {journal} {Phys. Rev. A}\ }\textbf {\bibinfo {volume} {103}},\ \bibinfo
  {pages} {022421} (\bibinfo {year} {2021})}\BibitemShut {NoStop}%
\end{thebibliography}%

\end{document}